\documentclass[journal,10pt,twocolumn]{IEEEtran}
\ifCLASSINFOpdf

\else

\fi

\usepackage{amsmath}
\usepackage{amssymb}
\usepackage{amsthm}
\usepackage{stackengine}
\usepackage{mathrsfs}
\usepackage{bbm}
\usepackage{mathtools}
\usepackage{bm}
\usepackage{amsfonts}
\usepackage{cases}
\usepackage{graphicx}
\usepackage{float}
\usepackage{multirow}
\usepackage{color}
\usepackage{algorithm}
\usepackage{algpseudocode}
\usepackage{graphicx}
\usepackage{subcaption}
\usepackage{epstopdf}
\usepackage{lipsum}
\graphicspath{{./}{Fig/}}

\algnewcommand\algorithmicswitch{\textbf{switch}}
\algnewcommand\algorithmiccase{\textbf{case}}
\algnewcommand\algorithmicassert{\texttt{assert}}
\algnewcommand\Assert[1]{\State \algorithmicassert(#1)}%
\algdef{SE}[SWITCH]{Switch}{EndSwitch}[1]{\algorithmicswitch\ #1\ \algorithmicdo}{\algorithmicend\ \algorithmicswitch}%
\algdef{SE}[CASE]{Case}{EndCase}[1]{\algorithmiccase\ #1}{\algorithmicend\ \algorithmiccase}%
\algtext*{EndSwitch}%
\algtext*{EndCase}%

\usepackage{xcolor}

\newtheorem{definition}{Definition}[section]
\newtheorem{lemma}{Lemma}[section]

\newtheorem{corollary}{Corollary}[section]

\newtheorem{theorem}{Theorem}[section]
\newtheorem{remark}{Remark}[section]
\newtheorem{example}{Example}[section]

\newtheorem{problem}{Problem}[section]
\begin{document}
\title{\huge Temporal Logic Trees for Model Checking and Control Synthesis of Uncertain Discrete-time Systems}
\author{Yulong Gao, Alessandro Abate,  Frank J. Jiang, Mirco Giacobbe, Lihua Xie, and Karl H. Johansson
\thanks{This work is supported by the Knut and Alice Wallenberg Foundation, the Swedish Strategic Research Foundation, the Swedish Research Council, and the  Wallenberg AI, Autonomous Systems and Software Program (WASP) funded by the Knut and Alice Wallenberg Foundation.}
\thanks{Y. Gao, F. J. Jiang, and K. H. Johansson are
with the Division of Decision and Control
Systems, Royal Institute of Technology, Stockholm, Sweden
        {\tt\small yulongg@kth.se, frankji@kth.se,  kallej@kth.se}}%
\thanks{A. Abate and M. Giacobbe are with the Department of Computer Science, University of Oxford, United Kingdom
        {\tt\small alessandro.abate@cs.ox.ac.uk, mirco.giacobbe@cs.ox.ac.uk}}%
\thanks{Y. Gao and L. Xie are with
the School of Electrical and Electronic Engineering, Nanyang Technological University, Singapore
        {\tt\small ygao009@ntu.edu.sg, elhxie@ntu.edu.sg}}
}

%\markboth{IEEE Transactions on Automatic Control,~Vol.~XX, No.~XX, XX~XXXX}%
%{GAO \MakeLowercase{\textit{et al.}}: Self-triggered control for constrained systems via reachability analysis}

\maketitle

\begin{abstract}
We propose algorithms for performing model checking and control synthesis for discrete-time uncertain  systems under linear temporal logic (LTL) specifications.
We construct temporal logic trees (TLT) from LTL formulae via reachability analysis. In contrast to automaton-based methods, the construction of the TLT is  abstraction-free for infinite systems, that is, we do not construct discrete abstractions of the infinite systems.
Moreover, for a given transition system and an LTL formula, we prove that there exist both a universal TLT and an existential TLT via minimal and maximal reachability analysis, respectively.
We show that the universal TLT is an underapproximation for the LTL formula and the existential TLT is an overapproximation. We provide sufficient conditions and necessary conditions to verify whether a  transition  system  satisfies  an  LTL  formula by using the TLT  approximations.
As a major contribution of this work,
for a controlled transition system and an LTL formula, we prove that  a controlled TLT can be constructed from the LTL formula via control-dependent reachability analysis.   Based on the controlled TLT, we design an online control synthesis algorithm, under which a set of feasible control inputs can be generated at each time step. We also prove that this algorithm is recursively feasible. We illustrate the proposed methods for both finite and infinite systems and highlight the generality and online scalability with two simulated examples.
\end{abstract}

%\begin{IEEEkeywords}

%\end{IEEEkeywords}

% Can use something like this to put references on a page
% by themselves when using endfloat and the captionsoff option.
\ifCLASSOPTIONcaptionsoff
  \newpage
\fi

%\ale{recommendation: add empty boxes at the end of remarks and examples, which otherwise are not clearly distinguishe from regular text.}

\section{Introduction}

In the recent past the integration of computer science and control theory has promoted the development of new areas such as embedded systems design~\cite{Astrom2013}, hybrid systems theory ~\cite{Tabuada2009book}, and, more recently,  cyber-physical systems~\cite{Alurbook2015}.  Given a model of a dynamical process and a specification (i.e., a description of desired properties), two fundamental problems arise:
\begin{itemize}
  \item \emph{model checking}: automatically verifying whether the behavior of the model  satisfies the given specification;
  \item \emph{formal control synthesis}:  automatically designing controllers (inputs to the system) so that the behavior of the model provably satisfies the given specification.
\end{itemize}
Both problems are of great interest in disparate and diverse applications, such as robotics, transportation systems, and safety-critical embedded system design.  However, they  are challenging problems when considering dynamical systems affected by uncertainty, and in particular uncertain infinite (uncountable) systems under complex, temporal logic specifications.
%Since such LTL formulae are defined over infinite trajectories, it is difficult to handle the uncertainties propagating along the trajectory.
In this paper, we provide solutions to the model checking and formal control synthesis problems, for discrete-time uncertain systems under linear temporal logic (LTL) specification.

\vspace{-0.2cm}
\subsection{Related Work}
In general, LTL formulae are expressive enough to capture many important properties, e.g., safety (nothing bad will ever happen), liveness (something good will eventually
happen), and more complex combinations of Boolean and temporal statements \cite{Baier2008}.

%By constructing the product automaton of the transition system and the negation of the LTL fomula, the model checking can also by solved by verifying whether there is some path of the product automaton eventually always visiting the negation of the set of accepting states in the equivalent automaton of the negation of the LTL formula.

In the area of formal verification, a dynamical process is by and large modeled as a  finite transition system.
A typical approach to both model checking and control synthesis for a finite transition system and an LTL formula leverages automata theory.
It is known that each LTL formula can be transformed to an equivalent automaton \cite{Vardi1996}.
The model checking problem can be solved by verifying whether the intersection of the trace set of the transition system and the set of accepted languages of the automaton expressing the negation of the LTL formula is empty, or not \cite{Baier2008}.
 The control synthesis problem can be solved by the following steps: (1) translate the LTL formula into a deterministic automaton; (2) build a ``product automaton'' between the transition system and the obtained  automaton; (3) transform the product automaton into a game \cite{Rabin1969}; (4) solve the game \cite{Emerson1985,Piterman2006,Horn2005}; and (5) map the solution into a control strategy. %If the  product automaton can be transformed into a deterministic B\"{u}chi automaton, then the Rabin game in steps (3)--(4)  is replaced with a B\"{u}chi game \cite{Thomas2002,Kloetzer2008}.

In recent years, the study of model checking and control synthesis for dynamical systems with continuous (uncountable) spaces, which extends the standard setup in formal verification, has attracted significant attention within the control community. This has enabled the formal control synthesis for  interesting properties, which are more complex than the usual control objectives such as stability and set invariance. In order to adapt automaton-based methods to infinite systems, \emph{abstraction} plays a central role in both model checking and control synthesis, which entails:
(1) to abstract an infinite system to a finite transition system; (2) to conduct automaton-based model checking or control synthesis for the finite transition system; (3) if a solution is found, to map it back to the infinite system; otherwise, to refine the finite transition system and repeat the steps above.

In order to show the correctness of the solution  obtained from the abstracted finite system over the infinite system, an equivalence or inclusion
 relation between the abstracted finite system and the infinite system needs to be established  \cite{Girard2007}. Relevant notions include (approximate) bisimulations  and simulations. These relations and their variants have been explored for  systems that are incrementally (input-to-state)
stable~\cite{Girard2010,Zamani2014}, or systems with similar properties~\cite{Zamani2012}.  Recent work \cite{Yu2019} shows that the condition of approximate simulation  can be relaxed to controlled globally
asymptotic or practical stability with respect to a given set for nonlinear  systems. We remark that such condition holds for only a small class of systems in practice.

Based on abstractions, the problem of model checking for infinite systems has been studied in~\cite{TabuadaHSCC2003,YordanovAuto2013}. In~\cite{TabuadaHSCC2003}, it is shown that  model checking for  discrete-time, controllable, linear systems from LTL formulae is decidable through an equivalent finite abstraction.
In~\cite{YordanovAuto2013}, the authors study the problem of verifying whether a linear system  with additive uncertainty from some initial states satisfies a fragment of LTL formulae, which can be transformed to a deterministic B\"{u}chi automaton. The key idea is to use a formula-guided method to construct and refine a finite system abstracted from the linear system and guarantee their equivalence.
Along the same line, the problem of control synthesis has also been widely studied for linear systems~\cite{YordanovTAC2012}, nonlinear systems~\cite{Meyer2019}, stochastic systems~\cite{Haesaert2018}, hybrid systems~\cite{Karaman2008}, and stochastic hybrid systems~\cite{Cauchi2019}. The applications of control synthesis under LTL specifications include  single-robot control in dynamic environments~\cite{Ulusoy2014}, multi-robot control~\cite{Guo2016}, and transportation control~\cite{Coogan2016}.

Beyond automata-based methods, alternative attempts have been made for specific model classes. Receding horizon methods  are used to design controllers under LTL for deterministic linear systems~\cite{Ding2014} and uncertain linear systems~\cite{Wongpiromsarn2012}.  The control of Markov decision processes under LTL is considered  in~\cite{Ding2014} and further applied to multi-robot coordination in~\cite{Schillinger2019}. Control synthesis for dynamical systems has been extended also to other specifications like signal temporal logic  (STL)~\cite{Lindemann2019}, and probabilistic computational tree logic (PCTL)~\cite{Kwiatkowska2011}.
Interested readers may refer to the tutorial paper~\cite{BeltaCDC2016} and the book~\cite{Beltabook2017} for  detailed discussions.

\vspace{-0.2cm}
\subsection{Motivations}
Although the last two decades have witnessed a great progress on model checking and control synthesis for infinite systems from both theoretical and practical perspectives, there are some inherent restrictions in the dominant automaton-based methods.

First, abstraction from infinite systems to finite systems suffers from the curse of dimensionality:
abstraction techniques usually  partition the state space, and transitions are constructed via  reachability analysis. The computational complexity  increases exponentially with the system dimension.
Many works are dedicated to improving the computational efficiency by using overapproximation for (mixed) monotone systems \cite{Coogan2016}, or by exploiting the structure of the uncertainty \cite{Cauchi2019}. However, another issue with abstraction techniques is that only systems with ``good properties'' (e.g., incremental stability, or smooth dynamics) might admit finite abstractions with guarantees, which limits their generality.

Second, there are few results for handling general LTL formulae when an infinite system comes with uncertainty (e.g., bounded disturbance, or additive noise). In most contributions on control synthesis of uncertain systems, fragments of LTL formulae (e.g., bounded LTL or \emph{co-safe} LTL) are usually taken into account \cite{Sessa2018,Hashimoto2019}. As mentioned before, the LTL formulae are defined over infinite trajectories and it is  difficult to control  uncertainties propagating along such trajectories. This restriction results from conservative over-approximation in the computation  of forward reachable sets, which is widely used for abstraction, and which leads to  information loss when used with automaton-based methods.

Third, current methods usually lack online scalability. In many applications, full \emph{a priori} knowledge of a specification cannot be obtained. For example,  consider an automated vehicle required to move from some initial position to some destination without colliding into any obstacle (e.g., other vehicles and pedestrians). Since the trajectories of other vehicles and pedestrians cannot be accurately  predicted, we cannot in advance define a specification that captures all the possibilities during the navigation process. Thus, offline design of automaton-based methods is significantly restricted.

Finally, the controller obtained from automaton-based methods usually only contains a single control policy.
In some applications, e.g., human-in-the-loop control~\cite{Inoue2018,GaoIFAC2020}, a set of feasible control inputs are needed to provide more degrees of freedom in the actual implementation. For example, \cite{Inoue2018} studies a control problem where humans are given a higher priority than the automated system in the decision making process. A controller is designed to provide a set of admissible control inputs with enough degrees of freedom to allow the human
operator to easily complete the task.

\vspace{-0.2cm}
\subsection{Contributions}

Motivated by the above, this paper studies LTL model checking and reachability-based control synthesis for discrete-time uncertain systems. There are many results for reachability analysis on infinite systems \cite{BlanchiniBook2007,LygerosAuto1999} and the computation of both forward and backward reachable sets has been widely studied  \cite{ChenTAC2018,AlthoffTAC2014,Mitchell2011}. The connection between STL and reachability analysis is studied in \cite{Chen2018}, which inspires our work. The main contributions of this paper are three-fold:

%BertsekasAuto1971, Rakovic2006
%\mig{(1) We provide a method for constructing a universal under-approximation of the satisfaction set of an LTL formula or,dually, for constructing an existential over-approximation of the violation set of the formula. We introduce a tree structure---which we call Temporal Logic Tree (TLT)---for representing the universal or alternatively the existential approximation over uncertain discrete-time dynamical systems. We show that universal and existential TLTs can be constructed via reachability analysis (Theorems~\ref{The:UniTLTLTL} and \ref{The:ExisTLTLTL}). The construction of TLTs is abstraction-free for infinite systems and admits online implementation (Section~\ref{Sec:Example}).}
(1) We construct tree structures from  LTL formulae via reachability analysis over dynamical models. We denote the  tree structure as a temporal logic tree (TLT). The connection between TLT and LTL is shown to hold for both uncertain finite and infinite models. The construction of the TLT is  abstraction-free for infinite systems and admits online implementation, as demonstrated in Section~\ref{Sec:Example}.
%than state-of-the-art methods.
More specifically, given a system and an LTL formula, we prove that both a universal TLT and an existential TLT can be constructed  for the LTL formula via minimal and maximal reachability analysis, respectively (Theorems~\ref{The:UniTLTLTL} and \ref{The:ExisTLTLTL}). We also show that the universal TLT is an underapproximation for the LTL formula and the existential TLT is an overapproximation for the LTL formula. Our formulation does not restrict the generality of LTL formulae.

(2)
We provide a method for model checking of discrete-time dynamical systems using TLTs. We provide sufficient conditions to verify whether a transition system satisfies an LTL formula by using universal TLTs for under-approximating the satisfaction set,
or alternatively using existential TLTs for  over-approximating the violation set (Theorem~\ref{The:SufMC}). Dually, we   provide necessary conditions by using existential TLTs for over-approximating the satisfaction set,
or alternatively using universal TLTs for  under-approximating the violation set (Theorem~\ref{The:NecMC}).

(3) As a core and novel contribution of this work, we detail an approach for online control synthesis for a controlled transition system to guarantee that the controlled system will satisfy the specified LTL formula. Given a control system and an LTL formula, we construct a controlled TLT  (Theorem~\ref{The:ConTLTLTL}). Based on the obtained TLT, we design an online control synthesis algorithm, under which a set of feasible control inputs is generated at each time step (Algorithm 3). We prove that this algorithm is recursively feasible (Theorem~\ref{The:RecFea}). We provide applications to show the scalability of our methods.

\vspace{-0.2cm}
\subsection{Organization}
The remainder of the paper is organized as follows. In Section~\ref{Sec:Preliminaries}, we define the notion of transition system, recall the problem of reachability analysis, and provide preliminaries on LTL specifications. In Section~\ref{Sec:TLT}, we introduce TLT structures and  show how to construct  a TLT from a given LTL formula.
%and their approximation relation.
In Section~\ref{Sec:MoChe}, we solve the LTL model checking problem through the constructed TLT. Section~\ref{Sec:ConSyn} solves the LTL control synthesis problem. %In Section~\ref{Sec:Discussion}, we discuss how our approach compares to others in the literature and delve into implementation aspects.
In Section~\ref{Sec:Example}, we illustrate
the effectiveness of our approaches with two numerical examples. In
Section~\ref{Sec:Conclusion}, we conclude the paper with a discussion about
our work and future directions.

%\textbf{Notation.} Let $\mathbb{N}$ denote the set of nonnegative integers and $\mathbb{R}$ denote the set of real numbers. For some $q,s \in \mathbb{N}$ and $q<s$, let $\mathbb{N}_{\geq q}$ and $\mathbb{N}_{[q,s]}$ denote the sets $\{r \in \mathbb{N}\mid r\geq q\}$ and $\{r \in \mathbb{N}\mid q \leq r\leq s\}$, respectively. %When $\leq$, $\geq$, $<$, and $>$ are applied to vectors, they are interpreted element-wise.
%The indicator function of a set $\mathbb{X}$ is denoted by $\mathbbm{1}_{\mathbb{X}}(x)$, that is, if $x\in  \mathbb{X}$,
%$\mathbbm{1}_{\mathbb{X}}(x)=1$  and otherwise, $\mathbbm{1}_{\mathbb{X}}(x)=0$.

\begin{figure}
\centering
\begin{subfigure}[t]{.22\textwidth}
	\includegraphics[width=\linewidth]{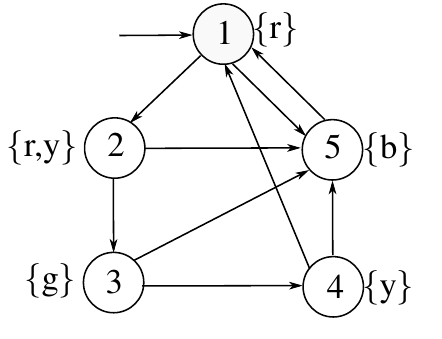}
\caption{}
\end{subfigure}\quad
\begin{subfigure}[t]{.1\textwidth}
	\includegraphics[width=\linewidth]{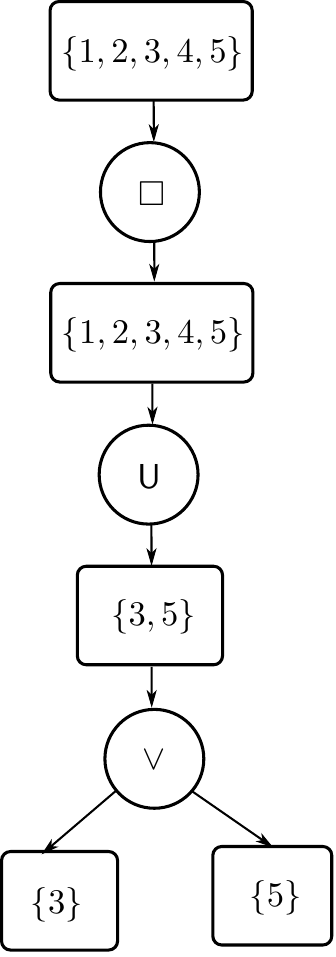}
\caption{}
\end{subfigure}
	\caption{\footnotesize  (a) A transition system illustrating a traffic light example. Labels are shown aside the states. The initial state is denoted by an incoming edge. (b) A TLT corresponding to an LTL formula $\varphi=\Box \diamondsuit  (g \vee b)$ for this system. Note that $\diamondsuit \varphi=\text{true} \cup \varphi$.}
	\label{Fig:trafficlight}
	\vspace{-0.5cm}
\end{figure}

\vspace{-0.2cm}
\section{Preliminaries}\label{Sec:Preliminaries}
This section will first introduce  transition systems and then recall reachability analysis and LTL.

\vspace{-0.2cm}
\subsection{Transition System}
\begin{definition}\label{Def:TS}
  A transition system $\mathsf{TS}$ is a tuple $\mathsf{TS}=(\mathbb{S},\rightarrow,\mathbb{S}_0,\mathcal{AP},L)$  consisting of
  \begin{itemize}
    \item a set $\mathbb{S}$ of states;
    \item a transition relation $\rightarrow\in \mathbb{S}\times \mathbb{S}$ \footnote{Here, the transition relation is not a functional relation, but instead for some state $x$, there may exist two different states $x'$ and $x''$ such that $x\rightarrow x'$ and $x\rightarrow x''$ hold. For notational simplicity, we use $\rightarrow\in \mathbb{S}\times \mathbb{S}$, rather than $\rightarrow\in \mathbb{S}\times 2^\mathbb{S}$. The same claim holds for the controlled transition systems in Section~\ref{Sec:ConSyn}.};
    \item a  set $\mathbb{S}_0\subseteq \mathbb{S}$ of initial states;
    \item a  set $\mathcal{AP}$ of atomic propositions;
    \item a labelling function $L: \mathbb{S} \rightarrow 2^{\mathcal{AP}}$.
  \end{itemize}
\end{definition}

\begin{definition}
  A transition system $\mathsf{TS}$ is said to be finite if $|\mathbb{S}|< \infty$ and $|\mathcal{AP}|< \infty$.
\end{definition}

\begin{definition}\label{Def:DTS}
  For $x\in \mathbb{S}$, the set $\mathsf{Post}(x)$ of direct successors of $x$ is defined by $\mathsf{Post}(x)=\{x'\in \mathbb{S}\mid x\rightarrow x'\}$. % and the set $\mathsf{Pre}(x)$ of direct predecessors of $x$ is defined by $\mathsf{Pre}(x)=\{x'\in \mathbb{S}\mid x'\rightarrow x\}$.
\end{definition}

\begin{definition}
  A transition system $\mathsf{TS}$ is said to be deterministic if $|\mathbb{S}_0|=1$ and    $|\mathsf{Post}(x)|=1$, $\forall x\in \mathbb{S}$.
\end{definition}

\begin{definition}
  (Trajectory \footnote{Notice that a trajectory $\bm{p}=x_0x_1\ldots x_kx_{k+1}\ldots$ is different from a \emph{trace}, which is the sequence of corresponding sets of atomic propositions, and is denoted by $L(x_0)L(x_1)\ldots L(x_k)L(x_{k+1})\ldots$.}) For a transition system $\mathsf{TS}$, an infinite trajectory $\bm{p}$ starting from $x_0$ is a sequence of states $\bm{p}=x_0x_1\ldots x_kx_{k+1}\ldots$ such that $\forall k\in \mathbb{N}$, $x_{k+1}\in \mathsf{Post}(x_k)$.
\end{definition}

 Denote by $\mathsf{Trajs}(x_0)$ the set of infinite trajectories starting from $x_0$. Let $\mathsf{Trajs}(\mathsf{TS})=\cup_{x\in\mathbb{S}_0}\mathsf{Trajs}(x)$.
For a trajectory $\bm{p}$, the $k$-th state is denoted by $\bm{p}[k]$, i.e., $\bm{p}[k]=x_k$ and the $k$-th prefix is denoted by $\bm{p}[..k]$, i.e., $\bm{p}[..k]= x_0\ldots x_k$.

%\begin{figure}
%\centering
%	\includegraphics[width=0.23\textwidth]{trafficlightpdf.pdf}
%	\caption{\footnotesize  A transition system illustrating a ``traffic light'' example. Labels are shown aside the states. The initial state is denoted by an incoming edge.}
%	\label{Fig:trafficlight}
%\end{figure}
\begin{example}\label{Exa:TraLig}
A traffic light can be red, green, yellow or
black (not working). The traffic light might stop working at any
time. After it has been repaired, it turns red. Initially, the light is
red. An illustration of such a traffic light is shown in Fig.~\ref{Fig:trafficlight}(a).  We can model the traffic light as a transition system $\mathsf{TS}=(\mathbb{S},\rightarrow,\mathbb{S}_0,\mathcal{AP},L)$:
  \begin{itemize}
    \item $\mathbb{S}=\{1,2,3,4,5\}$;
    \item $\rightarrow=\{(1, 2), (2, 3), (3, 4), (4, 1), (1, 5), (2, 5), (3, 5), \\ (4, 5), (5, 1)\}$;
    \item $\mathbb{S}_0=\{1\}$;
    \item $\mathcal{AP}=\{r, y, g, b\}$;
    \item $L = \{1\rightarrow\{r\},  2\rightarrow\{r,y\}, 3\rightarrow\{g\}, 4\rightarrow\{y\}, \\ 5\rightarrow\{b\}\}$. \qed
  \end{itemize}
\end{example}

\begin{remark}
We can rewrite the following  discrete-time autonomous system  as an infinite transition system:
\vspace{-0.15cm}
\begin{eqnarray*}
\mathsf{S}:
\begin{cases}
  x_{k+1}=f(x_k,w_k), \\
  y_k=g(x_k),
\end{cases}
\end{eqnarray*}
where $x_k\in\mathbb{R}^{n_x}$, $w_k\in \mathbb{R}^{n_w}$, $y_k\in 2^{\mathcal{O}}$, $f: \mathbb{R}^{n_x} \times \mathbb{R}^{n_w}\rightarrow \mathbb{R}^{n_x}$, and $g:\mathbb{R}^{n_x}\rightarrow 2^{\mathcal{O}}$. Here, $\mathcal{O}$ denotes the set of  observations. At each time instant $k$,  the disturbance $w_k$ belongs to a compact set $\mathbb{W}\subset \mathbb{R}^{n_w}$. Denote by $\mathsf{Ini}\subseteq\mathbb{R}^{n_x}$ the set of  initial states.
If $\mathcal{O}$ is finite, the system $\mathsf{S}$ can be rewritten as an infinite transition system $\mathsf{TS}_{\mathsf{S}}=(\mathbb{S},  \rightarrow,\mathbb{S}_0,\mathcal{AP},L)$ with
\begin{itemize}
  \item $\mathbb{S}=\mathbb{R}^{n_x}$;
   \item $\forall x,x'\in \mathbb{S}$,   $x \rightarrow x'$ if and only if there exists $w\in \mathbb{W}$ such that $x'=f(x,w)$;
  \item $\mathbb{S}_0=\mathsf{Ini}$;
  \item  $\mathcal{AP}=\mathcal{O}$;
  \item $L=g$.\qed
\end{itemize}
\end{remark}

\subsection{Reachability Analysis}
This subsection specifies the reachability analysis for a transition system $\mathsf{TS}$.
We first define the minimal reachable set and the maximal reachable set.
\begin{definition}\label{Def:miniReach}
   Consider a transition system $\mathsf{TS}$ and two sets $\Omega_1, \Omega_2 \subseteq \mathbb{S}$. The $k$-step minimal reachable set from $\Omega_1$ to $\Omega_2$ is defined as
   %\ale{in the next formula use $\mathbb{S}$ or $\mathbb{S}_0$? Similarly for later formulae in this subsection.}
   \begin{eqnarray*}
    &&\hspace{-0.7cm}\mathcal{R}^{\rm{m}}(\Omega_1,\Omega_2,k)=\Big\{x_0\in \mathbb{S}   \mid \forall \bm{p}\in \mathsf{Trajs}(x_0), \ {\rm s.t.},\\
    &&\hspace{-0.5cm}   \bm{p}[..k]=x_0\ldots x_{k}, \forall i\in\mathbb{N}_{[0,k-1]}, x_i\in \Omega_1, x_k\in\Omega_2 \Big\}.
    \end{eqnarray*}
    The minimal reachable set from $\Omega_1$ to $\Omega_2$ is defined as
    \begin{eqnarray*}
    &&\hspace{-0.5cm}\mathcal{R}^{\rm{m}}(\Omega_1,\Omega_2)= \bigcup_{k\in \mathbb{N}} \mathcal{R}^{\rm{m}}(\Omega_1,\Omega_2,k).
    \end{eqnarray*}
\end{definition}
\begin{lemma}\label{Lem:miniReach}
  For two sets $\Omega_1, \Omega_2 \subseteq \mathbb{S}$, define
   \begin{eqnarray*}
  &&\mathbb{Q}_{k+1}=\{x\in \Omega_1 \mid   \mathsf{Post}(x)\subseteq \mathbb{Q}_k\} \cup \mathbb{Q}_k , \\
  && \mathbb{Q}_0=\Omega_2.
  \end{eqnarray*}
  Then, $\mathcal{R}^{\rm{m}}(\Omega_1,\Omega_2)=\lim_{k\rightarrow \infty }\mathbb{Q}_k$.
\end{lemma}
\begin{proof}
From Definition~\ref{Def:miniReach}, it is easy to see that $$\mathbb{Q}_{k}=\bigcup_{i\in \mathbb{N}_{[0,k]}} \mathcal{R}^{\rm{m}}(\Omega_1,\Omega_2,i).$$
It follows from the Knaster-Tarski Theorem \cite{Tarski1955} that $\lim_{k\rightarrow \infty }\mathbb{Q}_k$ exists and is a  fixed point to the monotone function $F(\mathbb{P})=\{x\in \Omega_1 \mid   \mathsf{Post}(x)\subseteq \mathbb{P}\} \cup \mathbb{P}$.
  Thus, we have that $\mathcal{R}^{\rm{m}}(\Omega_1,\Omega_2)=\lim_{k\rightarrow \infty }\mathbb{Q}_k$.
\end{proof}

\begin{definition}
   Consider a transition system $\mathsf{TS}$ and two sets $\Omega_1, \Omega_2 \subseteq \mathbb{S}$. The $k$-step maximal reachable set from $\Omega_1$ to $\Omega_2$ is defined as %\ale{check from above}
    \begin{eqnarray*}
    &&\hspace{-0.7cm}\mathcal{R}^{\rm{M}}(\Omega_1,\Omega_2,k)=\Big\{x_0\in \mathbb{S}   \mid \exists \bm{p}\in \mathsf{Trajs}(x_0), \ {\rm s.t.}, \\ &&\hspace{-0.5cm} \bm{p}[..k]=x_0\ldots x_{N},
      \forall i\in\mathbb{N}_{[0,k-1]}, x_i\in \Omega_1, x_k\in\Omega_2 \Big\}.
    \end{eqnarray*}
    The maximal reachable set from $\Omega_1$ to $\Omega_2$ is defined as
    \begin{eqnarray*}
    &&\hspace{-0.5cm}\mathcal{R}^{\rm{M}}(\Omega_1,\Omega_2)= \bigcup_{k\in \mathbb{N}} \mathcal{R}^{\rm{M}}(\Omega_1,\Omega_2,k).
    \end{eqnarray*}
\end{definition}

\begin{lemma}\label{Lem:maxiReach}
  For two sets $\Omega_1, \Omega_2 \subseteq \mathbb{S}$, define
   \begin{eqnarray*}
  &&\mathbb{Q}_{k+1}=\{x\in \Omega_1 \mid   \mathsf{Post}(x)\cap \mathbb{Q}_k\neq \emptyset\} \cup \mathbb{Q}_k , \\
  && \mathbb{Q}_0=\Omega_2.
  \end{eqnarray*}
  Then, $\mathcal{R}^{\rm{M}}(\Omega_1,\Omega_2)=\lim_{k\rightarrow \infty }\mathbb{Q}_k$.
\end{lemma}
\begin{proof}
  Similar to the proof of Lemma~\ref{Lem:miniReach}.
\end{proof}

 We define the robust invariant set and the invariant set in the following.

\begin{definition}
  A set $\Omega_f\subseteq \mathbb{S}$ is said to be a robust invariant set of a transition system $\mathsf{TS}$ if for any $x\in \Omega_f$, $\mathsf{Post}(x)\subseteq \Omega_f$.
\end{definition}

\begin{definition}
For a set $\Omega \subseteq \mathbb{S}$, a set $\mathcal{RI}(\Omega)\subseteq \mathbb{S}$ is said to be the largest  robust invariant set  in  $\mathbb{S}$ if each robust invariant set  $\Omega_f\subseteq \Omega$ satisfies $\Omega_f\subseteq \mathcal{RI}(\Omega)$.
\end{definition}

\begin{lemma}\label{Lem:RobInv}
For a set $\Omega \subseteq\mathbb{S}$, define
  \begin{eqnarray*}
    &&\mathbb{Q}_{k+1}=\{x\in \mathbb{Q}_k \mid   \mathsf{Post}(x)\subseteq \mathbb{Q}_k\} , \\
    && \mathbb{Q}_0=\Omega.
    \end{eqnarray*}
Then,
 $\mathcal{RI}(\Omega )=\lim_{k\rightarrow \infty}\mathbb{Q}_{k}$.
\end{lemma}
\begin{proof}
 It follows again from the Knaster-Tarski Theorem~\cite{Tarski1955} that $\lim_{k\rightarrow \infty }\mathbb{Q}_k$ exists and it is a fixed point to the monotone function $F(\mathbb{P})=\{x\in\mathbb{P} \mid   \mathsf{Post}(x)\subseteq \mathbb{P}\} \cap \mathbb{P}$.
  Thus, we have that $\mathcal{RI}(\Omega )=\lim_{k\rightarrow \infty}\mathbb{Q}_{k}$.
\end{proof}

\begin{definition}
  A set $\Omega_f\subseteq \mathbb{S}$ is said to be an invariant set of a transition system $\mathsf{TS}$ if for any $x\in \Omega_f$, $\mathsf{Post}(x)\cap \Omega_f\neq \emptyset$.
\end{definition}

\begin{definition}
For a set $\Omega \subseteq \mathbb{S}$, a set $\mathcal{I}(\Omega)\subseteq \mathbb{S}$ is said to be the largest  invariant set in  $\mathbb{S}$ if each invariant set $\Omega_f\subseteq \Omega$ satisfies $\Omega_f\subseteq \mathcal{I}(\Omega)$.
\end{definition}

\begin{lemma}\label{Lem:Inv}
For a set $\Omega \subseteq\mathbb{S}$, define
  \begin{eqnarray*}
    &&\mathbb{Q}_{k+1}=\{x\in \mathbb{Q}_k \mid   \mathsf{Post}(x)\cap \mathbb{Q}_k\neq \emptyset\} , \\
    && \mathbb{Q}_0=\Omega.
    \end{eqnarray*}
Then,
 $\mathcal{I}(\Omega )=\lim_{k\rightarrow \infty}\mathbb{Q}_{k}$.
\end{lemma}
\begin{proof}
  Similar to the proof of Lemma~\ref{Lem:RobInv}.
\end{proof}

We can understand the reachable sets and invariant sets defined above as
  maps $\mathcal{R}^{\rm{m}}: 2^{\mathbb{S}}\times 2^{\mathbb{S}}\rightarrow 2^{\mathbb{S}}$, $\mathcal{R}^{\rm{M}}: 2^{\mathbb{S}}\times 2^{\mathbb{S}}\rightarrow 2^{\mathbb{S}}$, $\mathcal{RI}: 2^{\mathbb{S}}\rightarrow 2^{\mathbb{S}}$, and $\mathcal{I}: 2^{\mathbb{S}}\rightarrow 2^{\mathbb{S}}$, respectively.  In the following, we will refer to them as ``reachability operators''.

\subsection{LTL}
An LTL formula is defined over a finite set of atomic propositions $\mathcal{AP}$ and both logic and temporal operators. The syntax of LTL can be described as:
\begin{eqnarray*}\label{LTLdef}
\varphi ::= {\rm{true}} \mid a\in \mathcal{AP} \mid \neg  \varphi \mid \varphi_1 \wedge \varphi_2 \mid \bigcirc \varphi  \mid  \varphi_1 \mathsf{U} \varphi_2,
\end{eqnarray*}
where  $\bigcirc$ and $\mathsf{U}$ denote the ``next" and ``until" operators, respectively. By using the negation and  conjunction operators, we can define disjunction: $\varphi_1\vee \varphi_2=\neg (\neg \varphi_1 \wedge \neg \varphi_2)$. By employing the until operator, we can define: (1) eventually, $\diamondsuit \varphi=\text{true} \cup \varphi$; (2) always, $\Box \varphi= \neg  \diamondsuit  \neg \varphi$; and (3) weak-until, $\varphi_1 \mathsf{W}\varphi_2 = \varphi_1 \mathsf{U} \varphi_2 \vee \Box \varphi_1$.

% and (4) release, $\varphi_1 \mathsf{R}\varphi_2 =\neg(\neg\varphi_1 \mathsf{U} \neg\varphi_2)$.

\begin{definition}\label{LTLsemantic}
  (LTL semantics) For an LTL formula $\varphi$ and a trajectory $\bm{p}$, the satisfaction relation $\bm{p} \vDash \varphi$  is defined as
  \begin{eqnarray*}
  &&\bm{p} \vDash p \Leftrightarrow p \in L(x_0), \\
  && \bm{p} \vDash  \neg  p \Leftrightarrow p \notin L(x_0), \\
  && \bm{p} \vDash \varphi_1 \wedge \varphi_2 \Leftrightarrow \bm{p} \vDash \varphi_1 \wedge  \bm{p} \vDash \varphi_2, \\
   && \bm{p} \vDash \varphi_1 \vee  \varphi_2 \Leftrightarrow \bm{p} \vDash \varphi_1 \vee   \bm{p} \vDash \varphi_2, \\
  && \bm{p} \vDash \bigcirc \varphi  \Leftrightarrow \bm{p}[1..] \vDash \varphi,
 \\
  && \bm{p} \vDash \varphi_1 \mathsf{U} \varphi_2 \Leftrightarrow \exists j\in \mathbb{N} \ \text{s.t.} \begin{cases}
                \bm{p}[j..] \vDash \varphi_2, \\
                \forall i\in \mathbb{N}_{[0,j-1]}, \bm{p}[i..] \vDash \varphi_1,
              \end{cases}\\
  && \bm{p} \vDash \diamondsuit \varphi \Leftrightarrow \exists j\in \mathbb{N}, \ \text{s.t.} \ \bm{p}[j..] \vDash \varphi, \\
  && \bm{p} \vDash \Box \varphi \Leftrightarrow \forall j\in \mathbb{N}, \ \text{s.t.} \ \bm{p}[j..] \vDash \varphi, \\
  && \bm{p} \vDash \varphi_1 \mathsf{W} \varphi_2 \Leftrightarrow \begin{cases}
   \forall j\in \mathbb{N}, \bm{p}[j..] \vDash \varphi_1, \ \text{or} \\
          \exists j\in \mathbb{N} \ \text{s.t.} \begin{cases}
                \bm{p}[j..] \vDash \varphi_2, \\
                \forall i\in \mathbb{N}_{[0,j-1]}, \bm{p}[i..] \vDash \varphi_1.
              \end{cases}     \end{cases}
  %&& \bm{p} \vDash \varphi_1 \mathsf{R} \varphi_2 \Leftrightarrow \begin{cases}
%   \forall j\in \mathbb{N}, \bm{p}[j..] \vDash \varphi_2, \ \text{or} \\
%          \exists j\in \mathbb{N} \ \text{s.t.} \begin{cases}
%                \bm{p}[j..] \vDash \varphi_1, \\
%                \forall i\in \mathbb{N}_{[0,j]}, \bm{p}[i..] \vDash \varphi_2.
%              \end{cases}     \end{cases}
  \end{eqnarray*}
\end{definition}

\begin{definition}
  Consider a transition system $\mathsf{TS}$ and an LTL formula $\varphi$. The semantics of the universal form of  $\varphi$, denoted by $\forall \varphi$, is
  \begin{eqnarray*}
  x_0\vDash \forall \varphi \Leftrightarrow \forall \bm{p}\in \mathsf{Trajs}(x_0),  \bm{p}\vDash \varphi.
  \end{eqnarray*}
  The semantics of the existential form of  $\varphi$, denoted by $\exists \varphi$, is
  \begin{eqnarray*}
  x_0\vDash \exists \varphi \Leftrightarrow \exists \bm{p}\in \mathsf{Trajs}(x_0),  \bm{p}\vDash \varphi.
  \end{eqnarray*}
\end{definition}

\section{Temporal Logic Trees}\label{Sec:TLT}
This section will introduce the notion of TLT and establish a satisfaction relation between a trajectory and a TLT. Then, we construct TLTs from LTL formulae and discuss the approximation relation between them.

\subsection{Definitions}
\begin{definition}
  A TLT is a tree for which the next holds:
  \begin{itemize}
    \item each node is either a \emph{set} node, \emph{a subset of $\mathbb{S}$}, or an \emph{operator} node, \emph{from $\{\wedge,\vee,\bigcirc,\mathsf{U},\Box\}$}; %\mathsf{R}
    \item the root node and the leaf nodes are set nodes;
    \item if a set node is not a leaf node, its unique child  is an operator node;
    \item the children of any operator node are set nodes.
  \end{itemize}
\end{definition}

Next we define the complete path and the minimal Boolean fragment for a TLT. Minimal Boolean fragments play an important role when simplifying the TLT for model checking and control synthesis in the following.

\begin{definition}
 A complete path of a TLT is a sequence of nodes and edges from the root node to a leaf node. Any subsequence of a complete path is called a fragment of the complete path.
\end{definition}

\begin{definition}\label{Def:minBoolfrag}
  A minimal Boolean fragment of a complete path is one of the following fragments:
  \begin{itemize}
    \item[(i)] a fragment from the root node to the first Boolean operator node  ($\wedge$ or $\vee$) in the complete path;
    \item[(ii)] a segment from one Boolean operator node to the next Boolean operator node  in the complete path;
    \item[(iii)] a fragment from the last Boolean operator node of the complete path to the leaf node;
  \end{itemize}
\end{definition}

\begin{example}\label{Exa:TraLigTLT}
  Consider the traffic light in Example \ref{Exa:TraLig} and  the TLT in Fig.~\ref{Fig:trafficlight}(b), which corresponds to the LTL formula $\varphi=\Box \diamondsuit  (g \vee b)$ (the formal construction of a TLT from an LTL formula will be detailed in next subsection).  We encode one of the complete paths of this TLT  in the form of $\mathbb{X}_0 \Box\mathbb{X}_1\mathsf{U} \mathbb{X}_2 \vee \mathbb{X}_{3}$, where $\mathbb{X}_0=\mathbb{X}_1=\{1,2,3,4,5\}$, $\mathbb{X}_2=\{3,5\}$, and $\mathbb{X}_3=\{3\}$. For this complete path, the minimal Boolean fragments consist of $\mathbb{X}_0 \Box\mathbb{X}_1\diamondsuit \mathbb{X}_2 \vee$ and $\vee \mathbb{X}_{3}$, which correspond to  cases (i) and (iii) in Definition~\ref{Def:minBoolfrag}, respectively.  \qed
\end{example}

%\begin{figure}
%\centering
%	\includegraphics[width=0.12\textwidth]{TLTgbpdf.pdf}
%	\caption{\footnotesize  A TLT corresponding to an LTL formula $\varphi=\Box \diamondsuit  (g \vee b)$ for the traffic light in Example \ref{Exa:TraLig}. Note that $\diamondsuit \varphi=\text{true} \cup \varphi$.}
%	\label{Fig:TLTgb}
%\end{figure}

%\begin{figure}
%\centering
%	\subfigure{
%	\includegraphics[width=0.15\textwidth]{TLTofTrafficLightBETpdf.pdf}}
%	\caption{\footnotesize  The Boolean expression tree of the TLT in Fig.~\ref{Fig:TLTgb}}
%	\label{Fig:TLTgbBET}
%\end{figure}

We now define the satisfaction relation between a given trajectory and  a complete path of a TLT.
%\begin{definition}\label{Def:PrefixSaf}
% Consider a $k$-th prefix $\bm{p}[..k]=x_0x_1\ldots x_k$ and encode a fragment of a complete path in a TLT in the form of $\mathbb{X}_0\odot_1\mathbb{X}_1\odot_2\ldots \odot_{N_f} \mathbb{X}_{N_f}$ where $\mathbb{X}_0$ is the root node and  $N_f$ is the number of operators in the fragment,  $\mathbb{X}_i\subseteq \mathbb{S}$ for all $i\in \mathbb{N}_{[0,N_f]}$and  $\odot_{i}\in \{ \wedge,\vee,\bigcirc,\mathsf{U},\Box\}$  for all $i\in \mathbb{N}_{[1,N_f]}$.   The prefix $\bm{p}[..k]$ is said to satisfy this fragment if $x_0\in \mathbb{X}_0$,  $x_k\in \mathbb{X}_{N_f}$, and
% there exists a sequence of time steps $k_0k_1,\ldots,k_{N_f}$ with $k_i\leq \mathbb{N}$ for all  $i\in \mathbb{N}_{[0,N_f]}$ and $0\triangleq k_0\leq k_1\leq k_2\leq \ldots \leq k_{N_f}\leq k$ such that for all $i\in \mathbb{N}_{[0,N_f]}$,
% \begin{itemize}
%   \item if $\odot_i=\wedge$ or $\odot_i=\vee$, $x_{k_i}\in \mathbb{X}_{i-1}$ and $x_{k_i}\in \mathbb{X}_i$;
%   \item if $\odot_i=\bigcirc$, $x_{k_i-1}\in \mathbb{X}_{i-1}$ and $x_{k_i}\in \mathbb{X}_i$;
%   \item if $\odot_i=\mathsf{U}$, $x_{j}\in \mathbb{X}_{i-1}$, $\forall j\geq \mathbb{N}_{[k_{i-1},k_i-1]}$, and $x_{k_i}\in \mathbb{X}_{i}$;
%    \item if $\odot_i=\Box$, $x_{j}\in \mathbb{X}_i$, $\forall j\in \mathbb{N}_{[k_i,k]}$.
% \end{itemize}
%\end{definition}

\begin{definition}\label{Def:PathSaf}
 Consider a trajectory $\bm{p}=x_0x_1\ldots x_k\ldots$ and encode a complete path of a TLT in the form of $\mathbb{X}_0\odot_1\mathbb{X}_1\odot_2\ldots \odot_{N_f} \mathbb{X}_{N_f}$ where $N_f$ is the number of operators in the complete path,  $\mathbb{X}_i\subseteq \mathbb{S}$ for all $i\in \mathbb{N}_{[0,N_f]}$and  $\odot_{i}\in \{ \wedge,\vee,\bigcirc,\mathsf{U},\Box\}$  for all $i\in \mathbb{N}_{[1,N_f]}$.
 The trajectory $\bm{p}$ is said to satisfy this complete path if $x_0\in \mathbb{X}_0$ and
 there exists a sequence of time steps $k_0k_1,\ldots,k_{N_f}$ with $k_i\in \mathbb{N}$ for all  $i\in \mathbb{N}_{[0,N_f]}$ and $0\triangleq k_0\leq k_1\leq k_2\leq \ldots \leq k_{N_f}$ such that for all $i\in \mathbb{N}_{[0,N_f]}$,
 \begin{itemize}
   \item[(i)] if $\odot_i=\wedge$ or $\odot_i=\vee$, $x_{k_i}\in \mathbb{X}_{i-1}$ and $x_{k_i}\in \mathbb{X}_i$;
   \item[(ii)] if $\odot_i=\bigcirc$, $x_{k_i-1}\in \mathbb{X}_{i-1}$ and $x_{k_i}\in \mathbb{X}_i$;
   \item[(iii)] if $\odot_i=\mathsf{U}$, $x_{j}\in \mathbb{X}_{i-1}$, $\forall j\geq \mathbb{N}_{[k_{i-1},k_i-1]}$, and $x_{k_i}\in \mathbb{X}_{i}$;
    \item[(iv)] if $\odot_i=\Box$, $x_{j}\in \mathbb{X}_i$, $\forall j\geq k_i$.
 \end{itemize}
Consider a $k$-th prefix $\bm{p}[..k]=x_0x_1\ldots x_k$ from $\bm{p}$ and a fragment from the complete path in the form of $\mathbb{X}_0\odot_1\mathbb{X}_1\odot_2\ldots \odot_{N'_f} \mathbb{X}_{N'_f}$ where $N'_f\leq N_f$. The prefix $\bm{p}[..k]$ is said to satisfy this fragment if $x_0\in \mathbb{X}_0$,  $x_k\in \mathbb{X}_{N'_f}$, and
 there exists a sequence of time steps $k_0k_1,\ldots,k_{N'_f}$ with $k_i\in \mathbb{N}$ for all  $i\in \mathbb{N}_{[0,N'_f]}$ and $0\triangleq k_0\leq k_1\leq k_2\leq \ldots \leq k_{N'_f}\leq k$,  such that for all $i\in \mathbb{N}_{[0,N'_f]}$, (i)--(iii) holds and furthermore
 \begin{itemize}
    \item[(iv')] if $\odot_i=\Box$, $x_{j}\in \mathbb{X}_i$, $\forall j\in \mathbb{N}_{[k_i,k]}$.
 \end{itemize}

\end{definition}

\begin{definition}
  A time coding of a TLT is an assignment of  each operator node in the tree to a nonnegative integer.
\end{definition}

\begin{definition}\label{Def:TreeSaf}
Consider a trajectory $\bm{p}=x_0x_1\ldots x_k\ldots$ and a TLT.  The trajectory $\bm{p}$ is said to satisfy the TLT if there exists a time coding such that the output of Algorithm 1 is ${\rm true}$.
\end{definition}

The time coding indicates when the operators in the TLT are activated along a given trajectory.
Algorithm 1  provides a procedure to test if a trajectory satisfies a TLT under  a given time coding. The TLT is first transformed into a compressed tree, which is analogous to a  binary decision diagram (lines 1--2), through Algorithm 2. Then, we check if the trajectory satisfies each complete path of the TLT under the time coding (lines 3--9).  Finally, we backtrack the tree with output ${\rm true}$ or ${\rm false}$. If the output is ${\rm true}$,  the trajectory satisfies the TLT; otherwise, the trajectory does not  satisfy the TLT  under the given time coding.

Algorithm 2 aims to obtain a  tree in a compact form. Each minimal Boolean fragment is encoded according to Definition~\ref{Def:minBoolfrag}. The notation $\odot_{i}$ denotes the operator node and $N_f$ denotes the number of set nodes in the corresponding  minimal Boolean fragment.  We compress the sets in the minimal Boolean fragment to be a single set. The simplified tree consists of set nodes and Boolean operator nodes.

\begin{example}\label{Exa:TLTgb}
% Let us follow up  and \ref{Exa:TraLigBET}.
 From Definition~\ref{Def:PathSaf}, we can verify that the trajectory $\bm{p}=(1234)^{\omega}$ satisfies the complete path given in Example~\ref{Exa:TraLigTLT} by choosing $k_0=k_1=0$ and $k_2=k_3=2$. It follows from Definition~\ref{Def:TreeSaf}  that this trajectory  satisfies the corresponding TLT. \qed
\end{example}

\begin{algorithm}
\caption{TLT Satisfaction}
\hspace*{\algorithmicindent} \textbf{Input:} a trajectory $\bm{p}=x_0x_1\ldots x_k\ldots$, a TLT and a time coding  \\
 \hspace*{\algorithmicindent} \textbf{Output:} ${\rm true}$ or {\rm{false}};
\begin{algorithmic}[1]
\State construct a compressed tree via Algorithm 2 with input of
the TLT;
\State replace all set nodes of the compressed tree with ${\rm false}$;
\For {each complete path of the TLT}
\If{ $\bm{p}$ satisfies the complete path under the time coding}
\State set  the corresponding leaf node in the compressed tree with ${\rm true}$;
\Else
\State set the corresponding leaf node in the compressed tree  with ${\rm false}$;
\EndIf
\EndFor
\State backtrack the tree;
\State \textbf{return} the root node of the tree.
\end{algorithmic}
\end{algorithm}

\begin{algorithm}
\caption{Tree Compression}
\hspace*{\algorithmicindent} \textbf{Input:} a tree  \\
 \hspace*{\algorithmicindent} \textbf{Output:} a compressed tree
\begin{algorithmic}[1]
\For {each complete path of the tree}
\For {each minimal Boolean fragment}
\Switch{minimal Boolean fragment}
\Case{(i) in Definition~\ref{Def:minBoolfrag}}
\State encode the fragment in the form of $\mathbb{Y}_1\odot_1\ldots \odot_i\ldots \mathbb{Y}_{N_f}\odot_{N_f}$ with $\odot_{N_f}\in \{\wedge,\vee\}$;
\State replace the fragment with $\cup_{i=1}^{N_f}\mathbb{Y}_{i}\odot_{N_f}$;
\EndCase
\Case{(ii) in Definition~\ref{Def:minBoolfrag}}
\State encode the fragment in the form of $\odot_1\mathbb{Y}_1\odot_2\ldots \odot_{N_f} \mathbb{Y}_{N_f}\odot_{N_f+1}$ with $\odot_{1},\odot_{N_f+1}\in \{\wedge,\vee\}$;
\State replace the fragment with $\odot_{1} \cup_{i=1}^{N_f}\mathbb{Y}_{i}\odot_{N_f+1}$;
\EndCase
\Case{(iii) in Definition~\ref{Def:minBoolfrag}}
\State encode the fragment in the form of $\odot_1\mathbb{Y}_1\odot_2\ldots \odot_{N_f} \mathbb{Y}_{N_f}$ with $\odot_{1}\in \{\wedge,\vee\}$;
\State replace the fragment with $\odot_{1} \cup_{i=1}^{N_f}\mathbb{Y}_{i}$;
\EndCase
 \EndSwitch
 \State $\rhd$  \emph{$\odot_{i}$ denotes the operator node and $N_f$ denotes the number of set nodes in the   minimal Boolean fragment;}
 \EndFor
  \EndFor
  \State \textbf{return} the updated tree.
\end{algorithmic}
\end{algorithm}

\vspace{-0.2cm}
\subsection{Construction and Approximation of TLT}
We define the approximation relations between TLTs and LTL formulae as follows.

\begin{definition}
  A TLT is  said to be an under-approximation of an LTL formula $\varphi$ if all the trajectories that satisfy the TLT also satisfy $\varphi$.
\end{definition}

\begin{definition}
  A TLT is  said to be an over-approximation of an LTL formula $\varphi$,  if all the trajectories that satisfy $\varphi$ also satisfy the TLT.
\end{definition}

The following two theorems show how to  construct TLTs via reachability analysis for  the LTL formulae, and discuss their approximation relations.

\begin{theorem}\label{The:UniTLTLTL}
  For any transition system $\mathsf{TS}$ and any LTL formula $\varphi$,
  \begin{itemize}
    \item[(i)]  a TLT  can be constructed from the formula $\forall\varphi$ through  the reachability operators $\mathcal{R}^{\rm{m}}$ and $\mathcal{RI}$;
    \item[(ii)] this TLT
    %of $\varphi$
    is an under-approximation of $\varphi$.
  \end{itemize}
\end{theorem}
\begin{proof}
Here we provide a proof sketch.
See Appendix A for a detailed proof.

We prove the constructability by the following three steps: (1) we transform the given  LTL formula $\varphi$ into an equivalent LTL formula in the weak-until
positive normal form; (2) for each atomic proposition $a\in \mathcal{AP}$, we show that a TLT can be constructed from  $\forall a$ (or $\forall \neg a$); (3) we leverage induction
 to show the following: given LTL formulae $\varphi$, $\varphi_1$, and $\varphi_2$ in  weak-until
positive normal form, if TLTs can be constructed from $\forall \varphi$, $\forall \varphi_1$, and $\forall \varphi_2$, respectively, then TLTs can also be constructed through reachability operators $\mathcal{R}^{\rm{m}}$ and $\mathcal{RI}$ from  the formulae $\forall (\varphi_1 \wedge \varphi_2)$, $\forall (\varphi_1\vee \varphi_2)$, $\forall \bigcirc \varphi$, $\forall (\varphi_1 \mathsf{U} \varphi_2)$, and $\forall (\varphi_1 \mathsf{W} \varphi_2)$, respectively.  % ordia nextdia untildia

%The illustration diagrams of these constructions are shown in Fig.~\ref{Fig:TLTCons}.

We  follow a similar approach to  prove an under-approximation relation between the constructed TLT and the LTL formula. The under-approximation occurs due to the conjunction operator and the presence of branching in the transition system.
\end{proof}

Similarly, the following results hold.
\begin{theorem}\label{The:ExisTLTLTL}
  For any transition system $\mathsf{TS}$ and any LTL formula $\varphi$,
  \begin{itemize}
    \item [(i)] a TLT  can be constructed from the formula $\exists\varphi$ through the reachability operators $\mathcal{R}^{\rm{M}}$ and $\mathcal{I}$;
     \item [(ii)]  this TLT
     %of $\varphi$
     is an over-approximation of $\varphi$.
  \end{itemize}
\end{theorem}

\begin{proof}
  The proof of the first part is similar to that of Theorem~\ref{The:UniTLTLTL} by replacing the universal quantifier $\forall$ and   the reachability  operators  $\mathcal{R}^{\rm{m}}$ and $\mathcal{RI}$ with the existential quantifier $\exists$ and the operators $\mathcal{R}^{\rm{M}}$ and $\mathcal{I}$, respectively. Also, the proof of the second part is similar to that of Theorem~\ref{The:UniTLTLTL} by following the definition of the maximal reachability analysis.
\end{proof}

We call the constructed  TLT of $\forall\varphi$ the  \emph{universal TLT} of $\varphi$ and the  TLT of $\exists\varphi$ the  \emph{existential TLT} of $\varphi$. We remark that the constructed TLT  is not unique: this is because an LTL formula can have different equivalent expressions (e.g, normal forms). Despite this, the approximation relations between an LTL formula and the corresponding TLT still hold.

The following corollary shows that the approximation relation between TLTs and LTL formulae are tight for deterministic transition systems.
\begin{corollary}\label{Cor:TLTApproxdeter}
  For any deterministic transition system $\mathsf{TS}$ and any LTL formula $\varphi$,
   %if the same equivalent LTL formula of $\varphi$ in the weak-until positive normal form is used for the TLT construction,
  the universal TLT and the existential TLT of $\varphi$ are identical.
\end{corollary}
\begin{proof}
  If the system is deterministic, it follows from Lemmas~\ref{Lem:miniReach}--\ref{Lem:maxiReach} and Lemmas~\ref{Lem:RobInv}--\ref{Lem:Inv} that for any $\Omega_1, \Omega_2 \subseteq \mathbb{S}$ and $\Omega \subseteq \mathbb{S}$, $\mathcal{R}^{\rm{m}}(\Omega_1,\Omega_2)=\mathcal{R}^{\rm{M}}(\Omega_1,\Omega_2)$ and $\mathcal{RI}(\Omega)=\mathcal{I}(\Omega)$. Then, by the same construction procedure, we have that the constructed universal TLT  is the same as the  constructed existential TLT.
\end{proof}

\begin{remark}
Computation of reachable sets plays a central role in the construction of the TLT. The computation of reachable sets is not the focus of this paper. Interested readers are referred to relevant results   \cite{ChenTAC2018,AlthoffTAC2014,Mitchell2011} and associated computational tools, e.g., the multi-parametric toolbox \cite{MPT3} and the Hamilton-Jacobi toolbox \cite{Mitchell2005}.  \qed

%Rakovic2006
 %Note that our method and  abstraction-based methods build on the computation of reachable sets. Since there is no appropriate metric for measuring the computational complexity of reachable sets, particularly for infinite-state systems, we cannot provide a fair comparison to show that our method has a lower computational complexity than  abstraction-based methods. Through a direct use of reachability analysis to construct the TLT, we remark that our method is applicable for many systems and LTL formulae.

\end{remark}
\begin{example}
  Consider the traffic light in Example \ref{Exa:TraLig} and  the LTL formula $\varphi=\Box \diamondsuit  (g \vee b)$  in Example~\ref{Exa:TLTgb} again. We follow the proof of Theorem~\ref{The:UniTLTLTL} to show the correspondence between $\forall \varphi$ and the TLT in Fig.~\ref{Fig:trafficlight}(b):
  \begin{itemize}
  \item [(1)] the  universal TLT of $g$ is a single set node, i.e., $\{3\}$ and the  universal TLT of $b$ is also a single set node, i.e,, $\{5\}$;
  \item [(2)] the root node of the  universal TLT of $g \vee b$ is the union of $\{3\}$ and  $\{5\}$, i.e., $\{3,5\}$;
   \item [(3)] the root node of the  universal TLT of $\diamondsuit (g \vee b)$ is $\mathcal{R}^{\rm{m}}(\mathbb{S},\{3,5\})=\{1,2,3,4,5\}$;
    \item [(4)]   the root node of the  universal TLT of $\Box\diamondsuit (g \vee b)$ is
       $\mathcal{RI}(\{1,2,3,4,5\})=\{1,2,3,4,5\}$.
\end{itemize}

We can follow the same steps in the proof of Theorem~\ref{The:ExisTLTLTL} to construct the existential TLT of $\varphi$, which is the same as the universal TLT of $\varphi$ for the system in Example~\ref{Exa:TraLig}.\qed
\end{example}

%\begin{proposition}
%  For a transition system $\mathsf{TS}$ and an LTL formula $\varphi$,  $\mathsf{TS} \vDash \varphi$ only if $\mathbb{S}_0$ is a subset of the root node of the universal TLT of $\varphi$.
%\end{proposition}
%
%\begin{corollary}
%For a deterministic transition system $\mathsf{TS}$ and an LTL formula $\varphi$,  $\mathsf{TS} \vDash \varphi$ if and only if $\mathbb{S}_0$ is a subset of
%the root node of the equivalent TLT of $\varphi$.
%\end{corollary}

\section{Model Checking via TLT}\label{Sec:MoChe}
This section focuses on the  model checking problem.
\begin{problem}\label{Pro:ModChe}
   Consider a transition system $\mathsf{TS}$ and an LTL formula $\varphi$. Verify whether $\mathsf{TS} \vDash \varphi$, i.e., $\forall x_0 \in  \mathbb{S}_0$, $x_0\vDash \forall \varphi$.
\end{problem}

Thanks to the approximation relations between the TLTs and the LTL formulae, we obtain the following lemma.
\begin{lemma}\label{Lemma:TLTApproxnondeter}
  For any transition system $\mathsf{TS}$ and any LTL formula $\varphi$,
   \begin{itemize}
    \item [(i)] $x_0\vDash \forall \varphi$ if $x_0$ belongs to the root node of the universal TLT of $\varphi$;
     \item [(ii)] $x_0\vDash \exists \varphi$ only if $x_0$ belongs to the root node of the existential TLT of $\varphi$.
  \end{itemize}
\end{lemma}
\begin{proof}
The first result follows from  that the root node of the universal TLT is an under-approximation of  the satisfaction set of $\varphi$, as shown in Theorem~\ref{The:UniTLTLTL}. Dually, the second result follows from that the root node of the universal TLT is an over-approximation of  the satisfaction set of $\varphi$, shown in Theorem~\ref{The:ExisTLTLTL}.
\end{proof}

The next theorem provides two sufficient conditions for  solving Problem~\ref{Pro:ModChe}.

\begin{theorem}\label{The:SufMC}
For a transition system $\mathsf{TS}$ and an LTL formula $\varphi$,  $\mathsf{TS} \vDash \varphi$  \emph{if} one of the following conditions holds:
\begin{itemize}
  \item [(i)]   the initial state set $\mathbb{S}_0$ is a subset of the root node of the universal TLT for $\varphi$;
   \item [(ii)]  no initial state from $\mathbb{S}_0$ belongs to the root node of the existential TLT for $\neg\varphi$.
\end{itemize}
\end{theorem}
\begin{proof}
Condition (i) directly follows from the first result of Lemma~\ref{Lemma:TLTApproxnondeter}.
Let us next prove  condition (ii). It follows that
  \begin{eqnarray*}
  \mathsf{TS} \vDash \varphi \Leftrightarrow \forall \bm{p} \in  \mathsf{Trajs}(\mathsf{TS}), \bm{p}\vDash \varphi \Leftrightarrow \forall   \bm{p} \in  \mathsf{Trajs}(\mathsf{TS}), \bm{p}\nvDash \neg\varphi.
  \end{eqnarray*}
  From the second result of Lemma~\ref{Lemma:TLTApproxnondeter}, if $x_0$ does not belong to the root node of the existential TLT of $\neg\varphi$, we have $\bm{p}\nvDash\neg\varphi$, $\forall \bm{p}\in\mathsf{Trajs}(x_0)$. Thus, the condition (ii) is sufficient for verifying $\mathsf{TS} \vDash \varphi$.
\end{proof}

Similarly, we derive  two necessary conditions for solving the model checking problem.

\begin{theorem}\label{The:NecMC}
For a transition system $\mathsf{TS}$ and an LTL formula $\varphi$,  $\mathsf{TS} \vDash \varphi$ \emph{only if} one of the following conditions holds:
\begin{itemize}
  \item [(i)]   the initial state set $\mathbb{S}_0$ is a subset of the root node of the existential  TLT for $\varphi$;
   \item [(ii)]  no initial state from $\mathbb{S}_0$ belongs to the root node of the universal TLT for $\neg\varphi$.
\end{itemize}
\end{theorem}
\begin{proof}
Similar to Theorem~\ref{The:SufMC}.
\end{proof}

%Recall the definition of forward reachable sets from $\mathbb{S}_0$ and the equivalent construction of the TLT for an LTL formula, it yields that the fact that  the transition system  $\mathsf{TS}$ admits no path $\bm{p}$ such that $\bm{p} \vDash \neg\varphi$ is equivalent to the fact that there exists a $N\in \mathbb{N}$ such that $\forall N'\geq N$, with inputs of  $\mathcal{FR}(\mathbb{S}_0,N')$ and the existential TLT of $\neg\varphi$, the output of Algorithm 1 is 0. The proof is completed.

%The sufficiency of condition (i) is due to the under-approximation relation between the universal TLT and the LTL formula, as shown in Theorem~\ref{The:UniTLTLTL}, while the sufficiency of the condition (ii) results from the over-approximation of the existential TLT in Theorem~\ref{The:ExisTLTLTL} and the negation of the LTL formula.
%The condition (i) is in general computationally intractable, since it might be hard to test each infinite trajectory, even for simpler finite transition systems. Next, we leverage instead condition (ii) to design a practically useful model checking algorithm.

Notice that the approximation relations between the TLT and the LTL formula are tight for deterministic transition systems, as shown in Corollary~\ref{Cor:TLTApproxdeter}. In this case, the model checking problem can be tackled as follows.
\begin{corollary}
  For a deterministic transition system $\mathsf{TS}$ and an LTL formula $\varphi$,  $\mathsf{TS} \vDash \varphi$  if and only if the initial state set $\mathbb{S}_0$ is a subset of the root node of  the universal (or existential) TLT for $\varphi$.
\end{corollary}
\begin{proof}
  Follows from Corollary~\ref{Cor:TLTApproxdeter}.
\end{proof}

The conditions in Theorems~\ref{The:SufMC}--\ref{The:NecMC} imply that one can directly do model checking by using TLTs, as shown in the following example.
\begin{example}
 Let us continue to  consider the traffic light and  the LTL formula $\varphi=\Box \diamondsuit  (g \vee b)$.  Let us verify whether $\mathsf{TS}\vDash \varphi$ by using the above method. Since the unique initial state $x_0$ belongs to the root node of the universal TLT of $\varphi$ shown in Fig.~\ref{Fig:trafficlight}(b), it follows from  condition (i) in Theorem~\ref{The:SufMC} that $\mathsf{TS}\vDash \varphi$. Next, we show how to use  condition (ii) to verify that $\mathsf{TS}\vDash \varphi$. \\
First of all, we have $\neg \varphi=\diamondsuit \Box (\neg g \wedge \neg b)$. Following the proof of Theorem~\ref{The:ExisTLTLTL}, we construct the existential TLT of $\neg \varphi$:
\begin{itemize}
  \item [(1)] the  existential TLT of $\neg g$ is a single set node, i.e., $\{1,2,4,5\}$ and the  existential TLT of $\neg b$ is also a single set node, i.e,, $\{1,2,3,4\}$;
  \item [(2)] the root node of the  existential TLT of $\neg g \wedge \neg b$ is the intersection of $\{1,2,4,5\}$ and  $\{1,2,3,4\}$, i.e., $\{1,2,4\}$;
   \item [(3)] the root node of the  existential TLT of $\Box (\neg g \wedge \neg b)$ is $\mathcal{I}(\{2,3,4,5\})=\emptyset$.
\end{itemize}
As the  existential TLT of $\neg \varphi$ is the empty set $\emptyset$, this implies that  condition (ii)  in Theorem~\ref{The:SufMC} holds and thus $\mathsf{TS}\vDash \varphi$.
 \qed
\end{example}

\section{Control Synthesis via TLT}\label{Sec:ConSyn}
This section will show how to use the TLT to do control synthesis. Before that, we will introduce the notion of controlled transition system and recall the controlled reachability analysis.
\subsection{Controlled Transition System}
\begin{definition}\label{Def:CTS}
  A controlled transition system $\mathsf{CTS}$ is a tuple $\mathsf{CTS}=(\mathbb{S}, \mathbb{U}, \rightarrow,\mathbb{S}_0,\mathcal{AP},L)$  consisting of
  \begin{itemize}
    \item a set $\mathbb{S}$ of states;
    \item a set $\mathbb{U}$ of control inputs;
    \item a transition relation $\rightarrow\in \mathbb{S}\times \mathbb{U}\times \mathbb{S}$;
    \item a set $\mathbb{S}_0$ of initial states;
    \item a set $\mathcal{AP}$ of atomic propositions;
    \item a labelling function $L: \mathbb{S} \rightarrow 2^{\mathcal{AP}}$.
  \end{itemize}
\end{definition}

\begin{definition}
  A controlled transition system $\mathsf{CTS}$ is said to be finite if $|\mathbb{S}|< \infty$, $|\mathbb{U}|< \infty$, and  $|\mathcal{AP}|< \infty$.
\end{definition}

%\begin{definition}
%  For a controlled transition system $\mathsf{CTS}=(\mathbb{S}, \mathbb{U}, \rightarrow,\mathbb{S}_0,\mathcal{AP},L)$, an embedded transition system $\mathsf{ETS}$ is defined as $\mathsf{ETS}=(\mathbb{S}, \rightarrow^{\rm{e}},\mathbb{S}_0,\mathcal{AP},L)$, where the transition relation $\rightarrow^{\rm{e}}\in \mathbb{S}\times \mathbb{S}$ is defined
%  as for any two states $x,x'\in \mathbb{S}$, if there exists $u\in \mathbb{U}(s)$ such that $x \xrightarrow{u} x'$, $x\rightarrow^{\rm{e}} x'$.
%\end{definition}

\begin{definition}\label{Def:DTS}
  For $x\in \mathbb{S}$ and $u\in \mathbb{U}$, the set $\mathsf{Post}(x,u)$ of direct successors of $x$ under $u$ is defined by $\mathsf{Post}(x,u)=\{x'\in \mathbb{S}\mid x \xrightarrow{u} x'\}$.
\end{definition}

\begin{definition}
For $x\in \mathbb{S}$, the set $\mathbb{U}(x)$ of admissible control inputs at state $x$ is defined by $\mathbb{U}(x)=\{u\in \mathbb{U}\mid \mathsf{Post}(x,u)\neq\emptyset\}$.
\end{definition}

%\begin{definition}
%  A transition system $\mathsf{TS}$ is said to be deterministic if $|\mathbb{S}_0|=1$ an $|\mathsf{Post}(x,u)|=1$, $\forall x\in \mathbb{S}$ and $\forall u\in \mathbb{U}(x)$
%\end{definition}

\begin{definition}\label{Def:policy}
  (Policy) For a controlled transition system $\mathsf{CTS}$,
 a \emph{policy} $\bm{\mu}=u_0u_1\ldots u_k\ldots$ is a sequence of maps $u_k: \mathbb{S}\rightarrow \mathbb{U}$. Denote by $\mathcal{M}$ the set of all policies.
\end{definition}

\begin{definition}
  (Trajectory) For a controlled transition system $\mathsf{CTS}$, an infinite \emph{trajectory} $\bm{p}$ starting from $x_0$ under a policy $\bm{\mu}=u_0u_1\ldots u_k\ldots$  is a sequence of states $\bm{p}=x_0x_1\ldots x_k\ldots$ such that $\forall k\in \mathbb{N}$, $x_{k+1}\in \mathsf{Post}(x_k,u_k(x_k))$.
Denote by $\mathsf{Trajs}(x_0,\bm{\mu})$ the set of infinite trajectories starting from $x_0$ under $\bm{\mu}$.
\end{definition}

\begin{example}\label{Exa:CTL}
 A controlled transition system $\mathsf{CTS}=(\mathbb{S}, \mathbb{U}, \rightarrow,\mathbb{S}_0,\mathcal{AP},L)$ is shown in Fig.~\ref{Fig:controlexample}(a), where
 \begin{itemize}
    \item $\mathbb{S}=\{s_1,s_2,s_3,s_4\}$;
    \item $\mathbb{U}=\{a_1,a_2\}$;
    \item $\rightarrow=\{(s_1, a_1, s_2), (s_1, a_1, s_3), (s_2, a_1, s_2), (s_2, a_1, s_3), \\ (s_2, a_1, s_3), (s_2, a_2, s_4), (s_3, a_1, s_2), (s_3, a_2, s_3), \\ (s_4, a_1, s_2), (s_4, a_1, s_4)\}$;
    \item $\mathbb{S}_0=\{s_1\}$;
    \item $\mathcal{AP}=\{o_1,o_2,o_3\}$;
    \item $L = \{s_1\rightarrow \{o_1\},  s_2\rightarrow\{o_2\}, s_3\rightarrow\{o_3\}, s_4\rightarrow\{o_2\}\}$. \qed
  \end{itemize}
\end{example}

\begin{remark}  \label{SSec:disConSys}
We express the following  discrete-time uncertain control system as an infinite controlled transition system:
\begin{eqnarray}\label{Eq:disConSys}
\mathsf{CS}:
\begin{cases}
  x_{k+1}=f(x_k,u_k,w_k), \\
  y_k=g(x_k),
\end{cases}
\end{eqnarray}
where $x_k\in\mathbb{R}^{n_x}$ and $u_k\in\mathbb{R}^{n_u}$, $w_k\in \mathbb{R}^{n_w}$, $y_k\in 2^{\mathcal{O}}$, $f: \mathbb{R}^{n_x}\times \mathbb{R}^{n_u}\times \mathbb{R}^{n_w}\rightarrow \mathbb{R}^{n_x}$, and $g:\mathbb{R}^{n_x}\rightarrow 2^{\mathcal{O}}$. Here, $\mathcal{O}$ denotes the set of  observations. At each time instant $k$, the control input $u_k$ is constrained by a compact set $\mathbb{U}_{\mathsf{CS}}\subset \mathbb{R}^{n_u}$  and the disturbance $w_k$ belongs to a compact set $\mathbb{W}\subset \mathbb{R}^{n_w}$. Denote by $\mathsf{Ini}\subseteq\mathbb{R}^{n_x}$ the set of the initial states.
If the observation set $\mathcal{O}$ is finite,  $\mathsf{CS}$ can be rewritten as an infinite controlled transition system,  $\mathsf{CTS}_{\mathsf{CS}}=(\mathbb{S}, \mathbb{U}, \rightarrow,\mathbb{S}_0,\mathcal{AP},L)$ where
\begin{itemize}
  \item $\mathbb{S}=\mathbb{R}^{n_x}$;
  \item $\mathbb{U}=\mathbb{U}_{\mathsf{CS}}$;
   \item $\forall x,x'\in \mathbb{S}$ and $\forall u\in \mathbb{U}$,  $x \xrightarrow{u} x'$ if and only if there exists $w\in \mathbb{W}$ such that $x'=f(x,u,w)$;
  \item $\mathbb{S}_0=\mathsf{Ini}$;
  \item  $\mathcal{AP}=\mathcal{O}$;
  \item $L=g$. \qed
\end{itemize}
\end{remark}

  \begin{figure}
\centering
\begin{subfigure}[t]{.25\textwidth}
	\includegraphics[width=\linewidth]{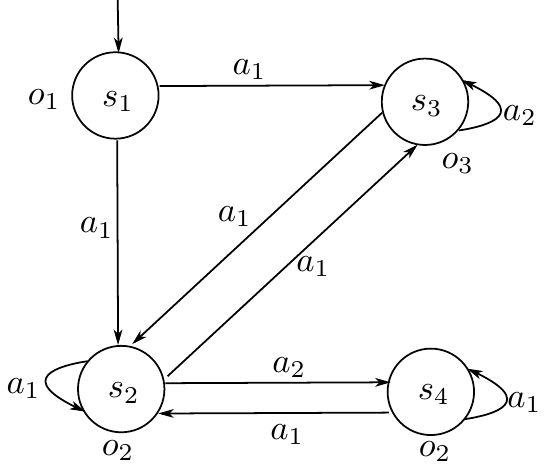}
\caption{}
\end{subfigure}\quad
\begin{subfigure}[t]{.09\textwidth}
	\includegraphics[width=\linewidth]{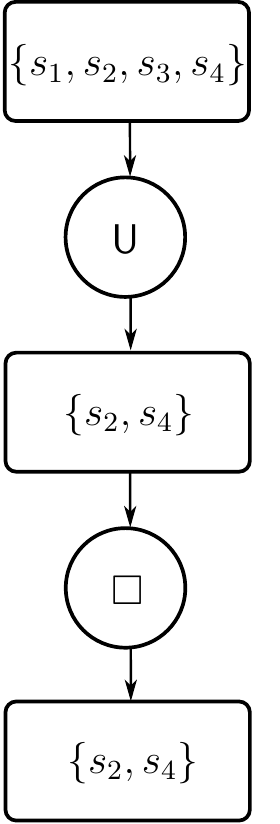}
\caption{}
\end{subfigure}
	\caption{\footnotesize (a) Graph description of a controlled transition system; (b) The controlled TLT of $\varphi= \diamondsuit  \Box o_2$ for the system.}
	\label{Fig:controlexample}
	\vspace{-0.2cm}
\end{figure}

 % \begin{figure}
%\centering
%	\includegraphics[width=0.26\textwidth]{controlexamplepdf.pdf}
%	\caption{\footnotesize  Graph description of a controlled transition system.}
%	\label{Fig:controlexample}
%\end{figure}

\vspace{-0.2cm}
\subsection{Controlled Reachability Analysis}\label{Subsec:ConReaAna}
%\frank{I think this section is too disconnected from Algorithm 7. Either make explicit reference to Algorithm 7 or make references in Algorithm 7 to this section. I think it would help the organization to connect the rechability analysis definitons directly to the control tree computation.}

This subsection will develop reachability analysis of a controlled transition system $\mathsf{CTS}$.
\begin{definition}
   Consider a controlled transition system $\mathsf{CTS}$ and two sets $\Omega_1, \Omega_2 \subseteq \mathbb{S}$. The $k$-step controlled  reachable set from $\Omega_1$ to $\Omega_2$ is defined as
   \begin{eqnarray*}
    &&\hspace{-0.7cm}\mathcal{R}^{\rm{c}}(\Omega_1,\Omega_2,k)=\Big\{x_0\in \mathbb{S}   \mid \exists \bm{\mu}\in \mathcal{M} \ {\rm s.t.}, \ \forall \bm{p}\in \mathsf{Trajs}(x_0, \bm{\mu}), \\
    &&\hspace{-0.5cm}   \bm{p}[..k]=x_0\ldots x_{k}, \forall i\in\mathbb{N}_{[0,k-1]}, x_i\in \Omega_1, x_k\in\Omega_2 \Big\}.
    \end{eqnarray*}
    The controlled reachable set from $\Omega_1$ to $\Omega_2$ is defined as
    \begin{eqnarray*}
    &&\hspace{-0.5cm}\mathcal{R}^{\rm{c}}(\Omega_1,\Omega_2)= \bigcup_{k\in \mathbb{N}} \mathcal{R}^{\rm{c}}(\Omega_1,\Omega_2,k).
    \end{eqnarray*}
\end{definition}

\begin{lemma}\label{Lemma:ConReaSet}
  For two sets $\Omega_1, \Omega_2 \subseteq \mathbb{S}$, define
   \begin{eqnarray*}
  &&\hspace{-0.5cm}\mathbb{Q}_{k+1}=\{x\in \Omega_1 \mid  \exists u\in \mathbb{U}(x), \mathsf{Post}(x,u)\subseteq \mathbb{Q}_k\} \cup \mathbb{Q}_k , \\
  &&\hspace{-0.5cm}\mathbb{Q}_0=\Omega_2.
  \end{eqnarray*}
  Then, $\mathcal{R}^{\rm{c}}(\Omega_1,\Omega_2)=\lim_{k\rightarrow \infty }\mathbb{Q}_k$.
\end{lemma}
\begin{proof}
  Similar to the proof of Lemma~\ref{Lem:miniReach}.
\end{proof}

\begin{definition}
  A set $\Omega_f\subseteq \mathbb{S}$ is said to be a robust controlled invariant set (RCIS) of a transition system $\mathsf{TS}$ if for any $x\in \Omega_f$, there exists $u\in \mathbb{U}(x)$ such that $\mathsf{Post}(x,u)\subseteq \Omega_f$.
\end{definition}

\begin{definition}
For a set $\Omega \subseteq \mathbb{S}$, a set   $\mathcal{RCI}(\Omega)\subseteq \mathbb{S}$ is said to be the largest  RCIS in  $\mathbb{S}$ if each RCIS $\Omega_f\subseteq \Omega$ satisfies $\Omega_f\subseteq \mathcal{RCI}(\Omega)$.
\end{definition}

\begin{lemma}\label{Lemma:RCIS}
For a set $\Omega \subseteq\mathbb{S}$, define
  \begin{eqnarray*}
    &&\mathbb{Q}_{k+1}=\{x\in \mathbb{Q}_k \mid  \exists u\in \mathbb{U}(x),  \mathsf{Post}(x,u)\subseteq \mathbb{Q}_k\} , \\
    && \mathbb{Q}_0=\Omega.
    \end{eqnarray*}
Then,
 $\mathcal{RCI}(\Omega )=\lim_{k\rightarrow \infty}\mathbb{Q}_{k}$.
\end{lemma}
\begin{proof}
   Similar to the proof of Lemma~\ref{Lem:RobInv}.
\end{proof}

 The definitions of controlled reachable sets and RCISs provide us a way to synthesize the feasible control set, which is detailed in Algorithm 4. In the following, we treat the maps $\mathcal{R}^{\rm{c}}: 2^{\mathbb{S}}\times 2^{\mathbb{S}}\rightarrow 2^{\mathbb{S}}$ and $\mathcal{RCI}: 2^{\mathbb{S}}\rightarrow 2^{\mathbb{S}}$ as the reachability operators.

\subsection{Construction and Approximation of TLT}
The next theorem shows how to construct a TLT from an LTL formula for a controlled  transition system and discusses its approximation relation.

\begin{theorem}\label{The:ConTLTLTL}
  For a controlled transition system $\mathsf{CTS}$ and any LTL formula $\varphi$, the following holds:
    \begin{itemize}
    \item[(i)]  a TLT  can be constructed from the formula $\varphi$ through  the reachability operators $\mathcal{R}^{\rm{c}}$ and $\mathcal{RCI}$;
    \item[(ii)] given an initial state $x_0$, if there exists a policy $\bm{\mu}$ such that $\bm{p}$ satisfies the constructed TLT, $\forall \bm{p} \in  \mathsf{Trajs}(x_0,\bm{\mu})$, then $\bm{p} \vDash \varphi$, $\forall \bm{p} \in  \mathsf{Trajs}(x_0,\bm{\mu})$.
  \end{itemize}
\end{theorem}
\begin{proof}
  The proof of the first part is similar to that of Theorem~\ref{The:UniTLTLTL} by replacing the reachability operators $\mathcal{R}^{\rm{m}}(\cdot)$ and  $\mathcal{RI}(\cdot)$ with $\mathcal{R}^{\rm{c}}(\cdot)$ and $\mathcal{RCI}(\cdot)$, respectively.

Similar to the  under-approximation of the universal TLT in  Theorem~\ref{The:UniTLTLTL}, we can show that the path satisfying the constructed TLT in the first part also satisfies the LTL formula.   Then, we can directly prove the second result.
\end{proof}

Let us call the constructed TLT of $\varphi$ in Theorem~\ref{The:ConTLTLTL}  the controlled TLT of $\varphi$.

\begin{remark}
Checking whether there exists a policy, such that all the resulting trajectories satisfy \emph{the obtained controlled TLT}, is in general a hard problem. A straightforward necessary condition is that $x_0$  belongs to the root node of the controlled TLT: however, this is neither a necessary nor a sufficient condition on the existence of a policy such that all the resulting trajectories satisfy \emph{the given LTL formula}. A (rather conservative) necessary condition for the latter case can be obtained by regarding the controlled TS as a non-deterministic transition system, and then applying Thereom~\ref{The:NecMC}.  \qed
\end{remark}

Next we will show how to construct the  controlled TLT through an example.
%The next proposition provides an approximate solution to Problem~\ref{Pro:LarSat}.
%\begin{proposition}
%  Consider a controlled transition system $\mathsf{CTS}$, an LTL formula $\varphi$.  If the root node of controlled TLT of $\varphi$ is nonempty, then it is an overapproximation of $\mathsf{Sat}(\mathsf{CTS},\varphi)$.
%\end{proposition}

%\begin{proposition}
% Consider a controlled transition system $\mathsf{CTS}$, its corresponding  embedded transition system $\mathsf{ETS}$, and an LTL formula $\varphi$. If
%  there exists a $N\in \mathbb{N}$ such that $\forall N'\geq N$, with inputs of  $\mathcal{FR}(\mathbb{S}_0,N')$ and the existential TLT of $\neg\varphi$ for $\mathsf{ETS}$, the output of Algorithm 1 is 0, then for all $x_0\in \mathbb{S}_0$,
%  there exists a policy $\bm{\mu}$ such that $\bm{p}\vDash \varphi$, $\forall \bm{p} \in  \mathsf{Trajs}(x_0,\bm{\mu})$
%\end{proposition}

\begin{example}\label{Exa:TLTCon}
Consider the controlled  transition system in Example~\ref{Exa:CTL}.
  For an LTL formula $\varphi= \diamondsuit  \Box  o_2$, we can follow the steps in the proof of Theorem~\ref{The:ConTLTLTL} to construct the controlled TLT of $\varphi$, as shown in Fig.~\ref{Fig:controlexample}(b). \qed
\end{example}

%\begin{figure}
%\centering
%	\includegraphics[width=0.1\textwidth]{controlexamTLTpdf.pdf}
%	\caption{\footnotesize  The controlled TLT of $\varphi= \diamondsuit  \Box o_2$ for the system shown in Fig.~\ref{Fig:controlexample}.}
%	\label{Fig:controlexamTLT}
%\end{figure}
\subsection{Control Synthesis Algorithm}
In this subsection, we solve  the following control synthesis problem.
%\begin{problem}\label{Pro:LarSat}
% (Largest Controlled Satisfying Region Problem)  Consider a controlled transition system $\mathsf{CTS}$ and an LTL formula $\varphi$. For an initial state $x_0\in \mathbb{S}_0$, find, if there exists, the largest set of states $\mathsf{Sat}(\mathsf{CTS},\varphi)$ such that $\forall x_0\in \mathsf{Sat}(\mathsf{CTS},\varphi)$, there exists a policy $\bm{\mu}$ such that  $\forall \bm{p}\in \mathsf{Trajs}(x_0,\bm{\mu})$,  $\bm{p} \vDash\varphi$.
%\end{problem}

\begin{problem}\label{Pro:ConSyn}
  Consider a controlled transition system $\mathsf{CTS}$ and an LTL formula $\varphi$. For an initial state $x_0\in \mathbb{S}_0$, find, whenever existing, a sequence of feedback control inputs $\bm{u}=u_0(x_0)u_1(x_1)\ldots u_k(x_k)\ldots$ such that the resulting trajectory $\bm{p}=x_0x_1\ldots x_k\ldots$ satisfies $\varphi$.
\end{problem}

\begin{remark}
Note that the objective of the above problem is not to find a policy $\bm{\mu}$, but a sequence of control inputs that depend on the measured state.  In general, synthesizing a policy $\bm{\mu}$ such that \textit{each trajectory} $\bm{p}\in \mathsf{Trajs}(x_0,\bm{\mu})$ satisfies $\varphi$ is computationally intractable for infinite systems. Instead, here we seek to find online a feasible control input at each time step, in a similar spirit to constrained control or receding horizon control. \qed
\end{remark}

Instead of directly solving Problem~\ref{Pro:ConSyn}, we consider the following related task, whose solution  is also a solution to Problem~\ref{Pro:ConSyn}, thanks to Theorem~\ref{The:ConTLTLTL}.

\begin{problem}\label{Pro:ConSynTLT}
Consider a controlled transition system $\mathsf{CTS}$ and an LTL formula $\varphi$. For an initial state $x_0\in \mathbb{S}_0$,  find, whenever existing, a sequence of control inputs $\bm{u}=u_0(x_0)u_1(x_1)\ldots u_k(x_k)\ldots$ such that  the resulting trajectory $\bm{p}=x_0x_1\ldots x_k\ldots$ satisfies the controlled TLT constructed from $\varphi$.
\end{problem}

 Algorithm 3 provides a solution to Problem~\ref{Pro:ConSynTLT}. In particular,  Algorithm 3 presents an online feedback control synthesis procedure, which consists of three steps: (1) control tree: replace each set node of the TLT with a corresponding control set candidate (Algorithm 4); (2) compressed control tree: compress the control tree as a new tree whose set nodes are control sets and whose operator nodes are conjunctions and disjunctions (Algorithm 2); (3) backtrack on the control sets through a bottom-up traversal over the compressed control tree (Algorithm 5). If the output of Algorithm 5 is  ${\rm NExis}$, there does not exist a feasible solution to Problem~\ref{Pro:ConSynTLT}. We remark that Algorithm 3 is implemented in a similar way to receding horizon control with the prediction horizon being one.

\begin{algorithm}
\caption{Online Feedback Control Synthesis via TLT}
\hspace*{\algorithmicindent} \textbf{Input:} an initial state $x_0\in \mathbb{S}_0$ and the controlled  TLT of an LTL formula $\varphi$   \\
 \hspace*{\algorithmicindent} \textbf{Output:}  ${\rm NExis}$ or $(\bm{u},\bm{p})$ with $\bm{u}=u_0u_1\ldots u_k\ldots$ and $\bm{p}=x_0x_1\ldots x_k\ldots$

\begin{algorithmic}[1]
\State initialize $k=0$;
%\If{$k=0$}
%\State construct a control tree via Algorithm 4, with as inputs $(x_k,k)$ and the TLT;
%\Else
\State construct a control tree via Algorithm 4, with  inputs $\bm{p}[..k]=x_0\ldots x_k$ and the TLT;
%\EndIf
\State construct a compressed tree via Algorithm 2, with input the control tree;
\State synthesize a control set $\mathbb{U}^{\varphi}_k(x_k)$ via Algorithm 5, with  input the compressed tree;
\If{$\mathbb{U}^{\varphi}_k(x_k)= \emptyset$}
\State stop and \textbf{return} ${\rm NExis}$;
\Else
\State choose $u_k\in \mathbb{U}^{\varphi}_k(x_k)$;
 \State implement $u_k$ and measure $x_{k+1}\in \mathsf{Post}(x_k,u_k)$;
\State update $k=k+1$ and go to Step 2;
 \EndIf
\end{algorithmic}
\end{algorithm}

\begin{algorithm}
\caption{Control Tree}
\hspace*{\algorithmicindent} \textbf{Input:}  $\bm{p}[..k]=x_0\ldots x_k$ and a TLT \\
 \hspace*{\algorithmicindent} \textbf{Output:} a control tree
\begin{algorithmic}[1]
\If{$k=0$ and $x_0\notin$ the root node of TLT}
\State \textbf{return} $\emptyset$
\Else
\For {each set node $\mathbb{X}$ of  TLT through a bottom-up traversal}
\If{$\bm{p}[..k]$ does not satisfy the fragment from the root node to $\mathbb{X}$}
\State $\rhd$ \emph{see Definition~\ref{Def:PathSaf}};
\State replace $\mathbb{X}$ with $\emptyset$;
\Else
%\If{$k\geq 1$ and $x_{k-1}\notin \mathbb{X}$ and $x_{k-1}\notin$ the parent of the parent of $\mathbb{X}$}
%\State replace $\mathbb{X}$ with $\emptyset$
%\Else
\If{$\mathbb{X}$ is leaf node}
\If{the parent of $\mathbb{X}$ is $\Box$}
\State replace $\mathbb{X}$ with  $\mathbb{U}_{\mathbb{X}}=\{u\in \mathbb{U}\mid \mathsf{Post}(x_k,u) \subseteq \mathcal{RCI}(\mathbb{X})\}$;
\Else
\State replace $\mathbb{X}$ with $\mathbb{U}$;
 \EndIf
\Else
\Switch{the child of  $\mathbb{X}$}
\Case{$\wedge$ (or $\vee$)}
\State replace $\mathbb{X}$ with $\mathbb{U}_{\mathbb{X}}=\cap_{i\in{\rm{CH}}} \mathbb{U}_{{\rm CH}, i}$ (or $\mathbb{U}_{\mathbb{X}}=\cup_{i\in{\rm{CH}}} \mathbb{U}_{{\rm CH}, i}$)
 \State $\rhd$ \emph{for each Boolean operator node,  ${\rm{CH}}$ collects its children and $\mathbb{U}_{{\rm CH}, i}$ is the corresponding control set for each child;}
%the intersection (or union) of the corresponding control set candidates of the children of $\wedge$ (or $\vee$) in the TLT;
\EndCase
\Case{$\bigcirc$}
\State replace $\mathbb{X}$ with  $\mathbb{U}_{\mathbb{X}}=\{u\in \mathbb{U}\mid \mathsf{Post}(x_k,u)\subseteq \mathbb{Y}\}$
\State $\rhd$ \emph{$\mathbb{Y}$ is the child of $\bigcirc$;}
\EndCase
\Case{$\mathsf{U}$ or $\Box$} %or $\mathsf{R}$}
\State replace $\mathbb{X}$ with  $\mathbb{U}_{\mathbb{X}}=\{u\in \mathbb{U}\mid \mathsf{Post}(x_k,u)\subseteq \mathbb{X}\}$;
\EndCase
 \EndSwitch
 \EndIf
 \EndIf
% \EndIf
 \EndFor
  \State \textbf{return} the updated tree as the control tree.
 \EndIf
\end{algorithmic}
\end{algorithm}

 Algorithm 4 aims to construct a control tree that enjoys the same tree structure as the input TLT. The difference is that all the  set nodes are replaced with the corresponding control set nodes. The correspondence is established as follows: (1) whether the initial state $x_0$ belongs to the  root node or not (lines 1--3); (2) whether the prefix  $\bm{p}[..k]$ satisfy the fragment from the root node to the set node (lines 5--7); %(3) whether the  state $x_{k-1}$ belongs to the same complete path as $x_k$ (lines 8--10);
 (4) whether or not the  set node is a leaf node (lines 9--14); (5) which operator the child of the  set node is (lines 16--24).

Algorithm 5 aims to backtrack a set by using the compressed tree. We update the parent of each Boolean operator node through a bottom-up traversal.
\begin{algorithm}
\caption{Set Backtracking}
\hspace*{\algorithmicindent} \textbf{Input:} a compressed tree  \\
 \hspace*{\algorithmicindent} \textbf{Output:} a set $\mathbb{U}_k^{\varphi}$
\begin{algorithmic}[1]
\For {each Boolean operator node of the compressed tree through a bottom-up traversal}
\Switch{Boolean operator}
\Case{$\wedge$}
\State replace its parent with  $\mathbb{Y}_{{\rm P}} \cup (\cap_{i\in{\rm{CH}}} \mathbb{Y}_{{\rm CH}, i})$;
\EndCase
\Case{$\vee$}
\State replace its parent with  $\mathbb{Y}_{{\rm P}} \cup (\cup_{i\in{\rm{CH}}} \mathbb{Y}_{{\rm CH}, i})$;
\EndCase
\EndSwitch
 \EndFor
 \State $\rhd$ \emph{for each Boolean operator node, $\mathbb{Y}_{{\rm P}}$ denotes its parent,  ${\rm{CH}}$ collects its children, and $\mathbb{Y}_{{\rm CH}, i}$ is the corresponding control set for each child;}
\State \textbf{return} the root node.
\end{algorithmic}
\end{algorithm}

Note that the construction of a control tree in Algorithm~4 is closely related to the controlled reachability analysis in Section~\ref{Subsec:ConReaAna}. In lines 12--13, the computation of control set follows from the definition of RCIS. In lines 22--23, the definition of one-step controlled reachable set is used to compute the control set. In lines 24--26, the control set is synthesized  from the definition of a controlled reachable set.

The following theorem shows that Algorithm 3 is  recursively feasible. This means that initial feasibility implies  future feasibility. This is an important property, particularly used  in receding horizon control.
\begin{theorem}\label{The:RecFea}
  Consider a controlled transition system $\mathsf{CTS}$, an LTL formula $\varphi$, and an initial state $x_0\in \mathbb{S}_0$. Let $x_0$ and the controlled TLT of $\varphi$ be the inputs of Algorithm 3.
If there exists a policy $\bm{\mu}$ such that $\bm{p}$ satisfies the controlled TLT of $\varphi$, $\forall \bm{p} \in  \mathsf{Trajs}(x_0,\bm{\mu})$, then
  \begin{itemize}
    \item[(i)]  the control set $\mathbb{U}^{\varphi}_k(x_k)$ (line 8 of  Algorithm 3)  is nonempty for all $k\in \mathbb{N}$;
    \item[(ii)]  at each time step $k$,  there exists at least one trajectory $\bm{p}$ with  prefix $\bm{p}[..k+1]=x_0\ldots x_kx_{k+1}$      under some policy such that  $\bm{p}$ satisfies the controlled TLT of $\varphi$,  $\forall u_k\in \mathbb{U}^{\varphi}_k(x_k)$ and $\forall x_{k+1}\in  \mathsf{Post}(x_k,u_k)$.
  \end{itemize}
\end{theorem}
\begin{proof}
The proof follows from the construction of the set $\mathbb{U}^{\varphi}_k(x_k)$ in Algorithm~4 and the operations in Algorithms~2 and~5, and the definitions of controlled reachable sets and RCIS.  If there exists a policy $\bm{\mu}$ such that $\bm{p}$ satisfies the controlled TLT of $\varphi$, $\forall \bm{p} \in  \mathsf{Trajs}(x_0,\bm{\mu})$, we have that Algorithm 3 is feasible at each time step $k$, which implies that $\mathbb{U}^{\varphi}_k(x_k)\neq \emptyset$. Furthermore, from Algorithm 4, each element in $\mathbb{U}^{\varphi}_k(x_k)$ guarantees the one-step ahead feasibility for all realizations of  $x_{k+1}\in  \mathsf{Post}(x_k,u_k)$, which implies the result~(ii).
\end{proof}

Theorem~\ref{The:RecFea} implies that if there exists a policy  such that all the resulting trajectories satisfy  the controlled TLT built from $\varphi$, then Algorithm 3 is always feasible at each time step in two senses: (1) the control set $\mathbb{U}^{\varphi}_k(x_k)$ is nonempty; and (2) there always exists a feasible policy such that the trajectories with the realized prefix satisfy the controlled TLT.

%\begin{remark}
%Note that from Theorem~\ref{The:ConTLTLTL}, the controlled TLT is an underapproximation  of the LTL formula. Thus, even though we cannot find a feasible solution to Problem~\ref{Pro:ConSyn} by  Algorithm 6, we cannot claim that there does not exist a policy $\bm{\mu}$ for the given LTL formula.
%\end{remark}

\begin{remark}
In Algorithm~3, the integration of Algorithms 2, 4, and 5 can be interpreted as a feedback control law. This control law is a set-valued map $\mathbb{S}^{k+1}\rightarrow 2^{\mathbb{U}}$ at time step $k$.
%, and in $\mathbb{S}\times \mathbb{S}\rightarrow 2^{\mathbb{U}}$ at time step $k$ with $k\geq1$.
Given the prefix $\bm{p}[..k]=x_0\ldots x_k$, this map collects all the feasible control inputs such that the state can move along the TLT from $\bm{p}[..k]$.
 %In this sense, Algorithm~3 is implemented under this stationary set-valued map, of which the product generates a policy space that contains the policy space $\mathcal{M}$ in Definition~\ref{Def:policy}.}
 \qed
\end{remark}

\begin{remark}
Note that  to implement Algorithm 3, we do not need to first check for the existence of a policy  for  the  controlled  TLT. The fact that a non-empty control set is  synthesized by Algorithm 3 at each time step is necessary for the existence of the policy for the controlled TLT. We use the existence of the policy as a-priori condition for proving the recursive feasibility of Algorithm 3 in Theorem~\ref{The:ConTLTLTL}.  \qed
\end{remark}

\begin{example}
Let us continue to consider the controlled  transition system in Example~\ref{Exa:CTL} and the LTL formula $\varphi= \diamondsuit  \Box  o_2$ in Example~\ref{Exa:TLTCon}. Implementing Algorithm 3, we obtain Table~I. We can see that at each time step, we can synthesize a nonempty feedback control set.  One  realization is $s_1 \xrightarrow{a_1} s_3 \xrightarrow{a_2} s_3 \xrightarrow{a_1} s_2 \xrightarrow{a_1} s_3\xrightarrow{a_1} s_2 \xrightarrow{a_2} s_4 \xrightarrow{a_1} s_2 \xrightarrow{a_2} s_4 \cdots$, of which the trajectory $\bm{p}=s_1s_3s_3s_2s_3s_2(s_4s_2)^{\omega}$ satisfies both the controlled TLT and the formula $\varphi$.

In this example, Algorithm 3 is recursively feasible since we can verify that the condition in Theorem~\ref{The:RecFea} holds. That is, there exists a policy such that all the  resulting  trajectories  satisfy  the  controlled  TLT: a feasible stationary  policy is  $\bm{\mu}=uu\cdots$, where $u: \mathbb{S}\rightarrow \mathbb{U}$ with $u(s_1)=a_1, u(s_2)=a_2, u(s_3)=a_1$, and $u(s_4)=a_1$. Under this policy, there are two possible trajectories, $\bm{p}=s_1s_3(s_2s_4)^{\omega}$ and $\bm{p}=s_1(s_2s_4)^{\omega}$, both of which satisfy the controlled TLT and the LTL formula $\varphi$.
 \qed
\end{example}

\begin{table}[t]
 \label{Table:ImpAlg5}
 \centering
   \caption{Online implementation under Algorithm 3}
 \begin{tabular}{c|c|c|c}
Time $k$ & State $x_k$ & Control set $\mathbb{U}^{\varphi}_k(x_k)$ & Control input $u_k$ \\
 \hline
$0$ &  $s_1$ & $\{a_1\}$ & $a_1$  \\
$1$ &  $s_3$ & $\{a_1,a_2\}$ & $a_2$  \\
$2$ &  $s_3$ & $\{a_1,a_2\}$ & $a_1$  \\
$3$ &  $s_2$ & $\{a_1,a_2\}$ & $a_1$  \\
$4$ &  $s_3$ & $\{a_1,a_2\}$ & $a_1$  \\
$5$ &  $s_2$ & $\{a_1,a_2\}$ & $a_2$  \\
$6$ &  $s_4$ & $\{a_1\}$ & $a_1$  \\
$7$ &  $s_2$ & $\{a_1,a_2\}$ & $a_2$  \\
$\vdots$ &  $\vdots$ & $\vdots$ & $\vdots$
\end{tabular}
 \end{table}

\section{Examples}\label{Sec:Example}
\subsection{Obstacle Avoidance}
Following the example of obstacle avoidance for double integrator  in \cite{Gol2014}, we consider the following dynamical system
\begin{eqnarray}\label{Eq:doubleintegrator}
x_{k+1}=\left[\begin{array}{ccccccc}
1 & 0.2 \\
0 & 1
\end{array}\right]x_k+ \left[\begin{array}{ccccccc}
0.1 \\
0.2
\end{array}\right]u_k+w_k.
\end{eqnarray}
Different from \cite{Gol2014}, we choose the smaller sampling time of $0.2$ second and take into account the disturbance $w_k$. We consider the same scenario as in \cite{Gol2014}, as shown in Fig.~\ref{Fig:example1_scenariopdf}(a). The working space is $\mathbb{X}=\{z\in \mathbb{R}^2\mid [-10, -10]^T \leq z\leq [2, 2]^T\}$, the control  constraint set is $\mathbb{U}=\{z\in \mathbb{R}\mid -2 \leq z\leq 2\}$, and the disturbance set is  $\mathbb{W}=\{z\in \mathbb{R}^2\mid [-0.05, -0.05]^T \leq z\leq [0.05, 0.05]^T\}$. In Fig.~\ref{Fig:example1_scenariopdf}(a), the obstacle regions are $\mathbb{O}_1=\{z\in \mathbb{R}^2\mid [-10, -10]^T \leq z\leq [-5, -5]^T\}$ and $\mathbb{O}_2=\{z\in \mathbb{R}^2\mid [-5, -4]^T \leq z\leq [2, -3]^T\}$, the target region is  $\mathbb{T}=\{z\in \mathbb{R}^2\mid [-0.5, -0.5]^T \leq z\leq [-0.5, -0.5]^T\}$, and two visiting regions are $\mathbb{A}=\{z\in \mathbb{R}^2\mid [-6, 1]^T \leq z\leq [-5, 2]^T\}$ and $\mathbb{B}=\{z\in \mathbb{R}^2\mid [-5, -3]^T \leq z\leq [-4, -2]^T\}$.

\begin{figure}
\centering
\begin{subfigure}[t]{.23\textwidth}
	\includegraphics[width=\linewidth]{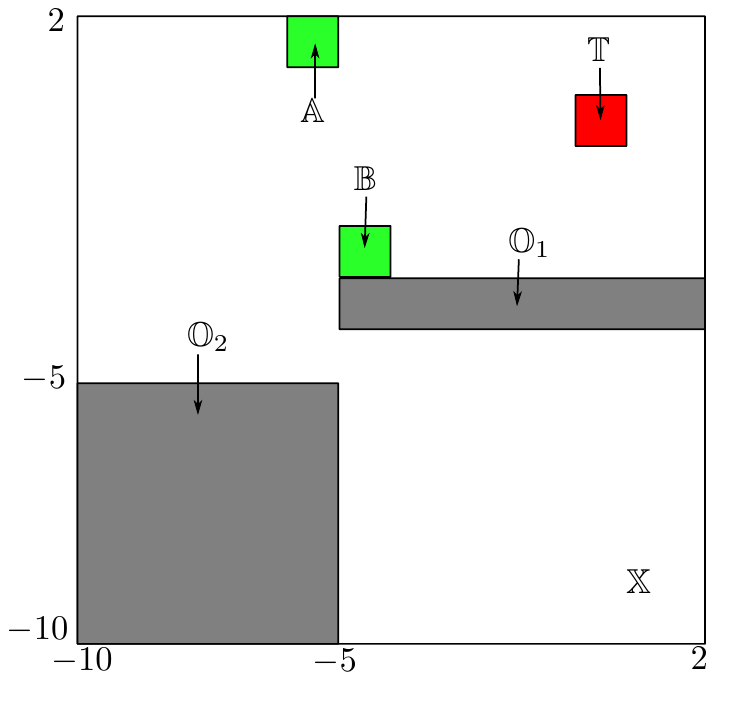}
\caption{}
\end{subfigure}
\begin{subfigure}[t]{.22\textwidth}
	\includegraphics[width=\linewidth]{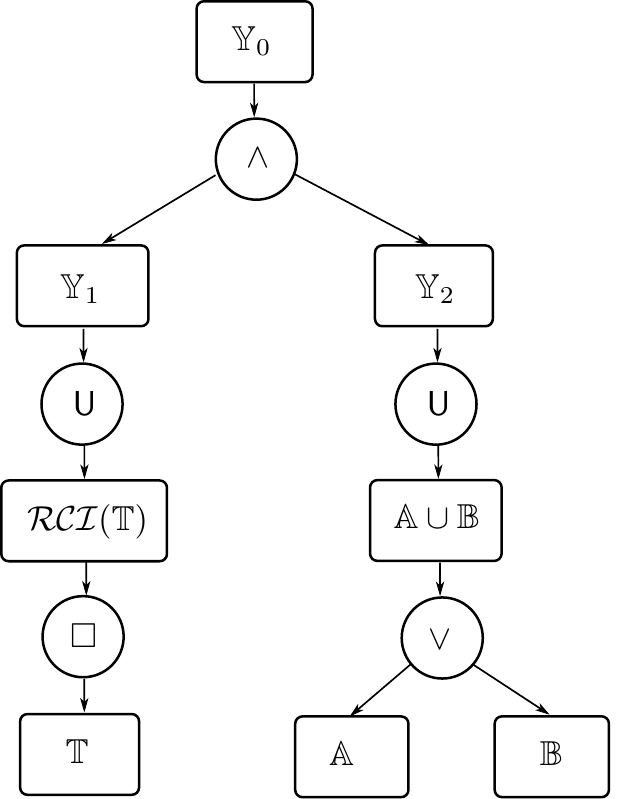}
\caption{}
\end{subfigure}
	\caption{\footnotesize  (a) Scenario of Example 1; (b) The controlled TLT for the LTL formula $\varphi= ((a_1 \wedge \neg a_2 \wedge \neg a_3) \mathsf{U} \Box a_6) \wedge (\neg a_6 \mathsf{U} (a_4 \vee a_5) )$ in Example 1, where $\mathbb{Y}_0=\mathbb{Y}_1\cap \mathbb{Y}_2$, $\mathbb{Y}_1=\mathcal{R}^{\rm c}(\mathbb{X}\setminus(\mathbb{O}_1\cup\mathbb{O}_2),\mathcal{RCI}(\mathbb{T}))$, and $\mathbb{Y}_2=\mathcal{R}^{\rm c}(\mathbb{X}\setminus \mathbb{T},\mathbb{A}\cup \mathbb{B})$.}
	\label{Fig:example1_scenariopdf}
\end{figure}

\begin{figure}
\hspace{-0.5cm}
	\includegraphics[width=0.55\textwidth]{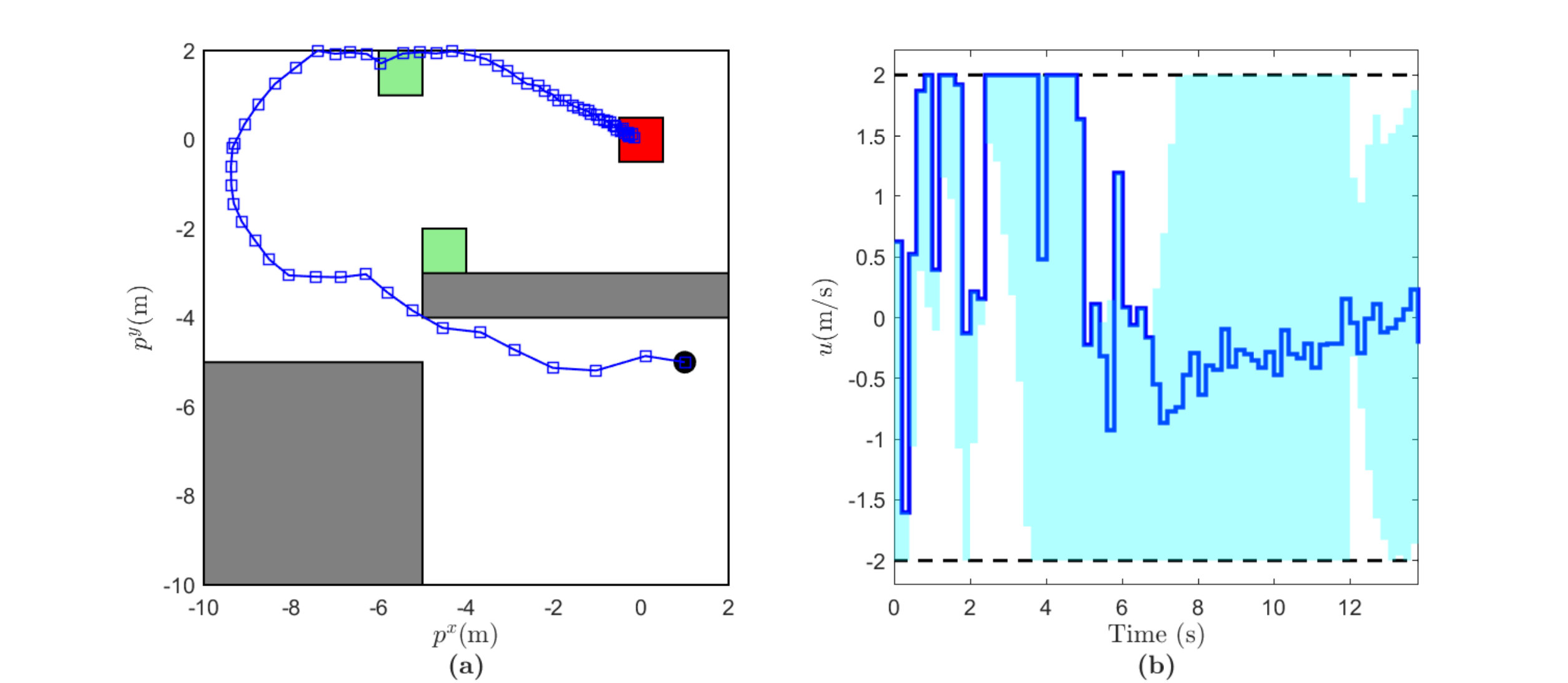}
	\caption{\footnotesize  Trajectories starting from the initial states $[1,-5]^T$: (a) state trajectory; (b) control trajectory together with control sets.}
	\label{Fig:Exam1_case1}
\end{figure}

\begin{figure}
\hspace{-0.5cm}
	\includegraphics[width=0.55\textwidth]{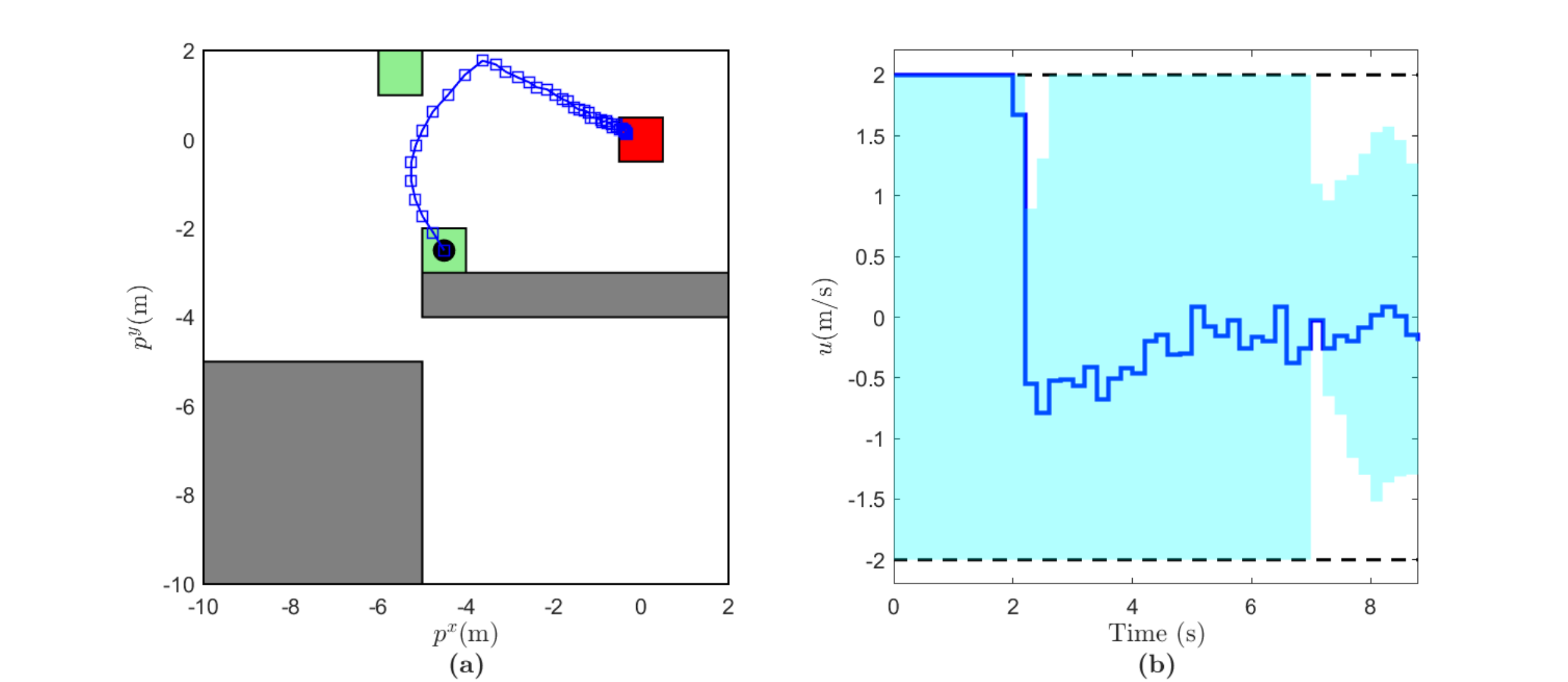}
	\caption{\footnotesize Trajectories starting from the initial states $[-4.5,-2.5]^T$: (a) state trajectory; (b) control trajectory.}
	\label{Fig:Exam1_case2}
\end{figure}

\begin{figure}
\hspace{-0.5cm}
	\includegraphics[width=0.55\textwidth]{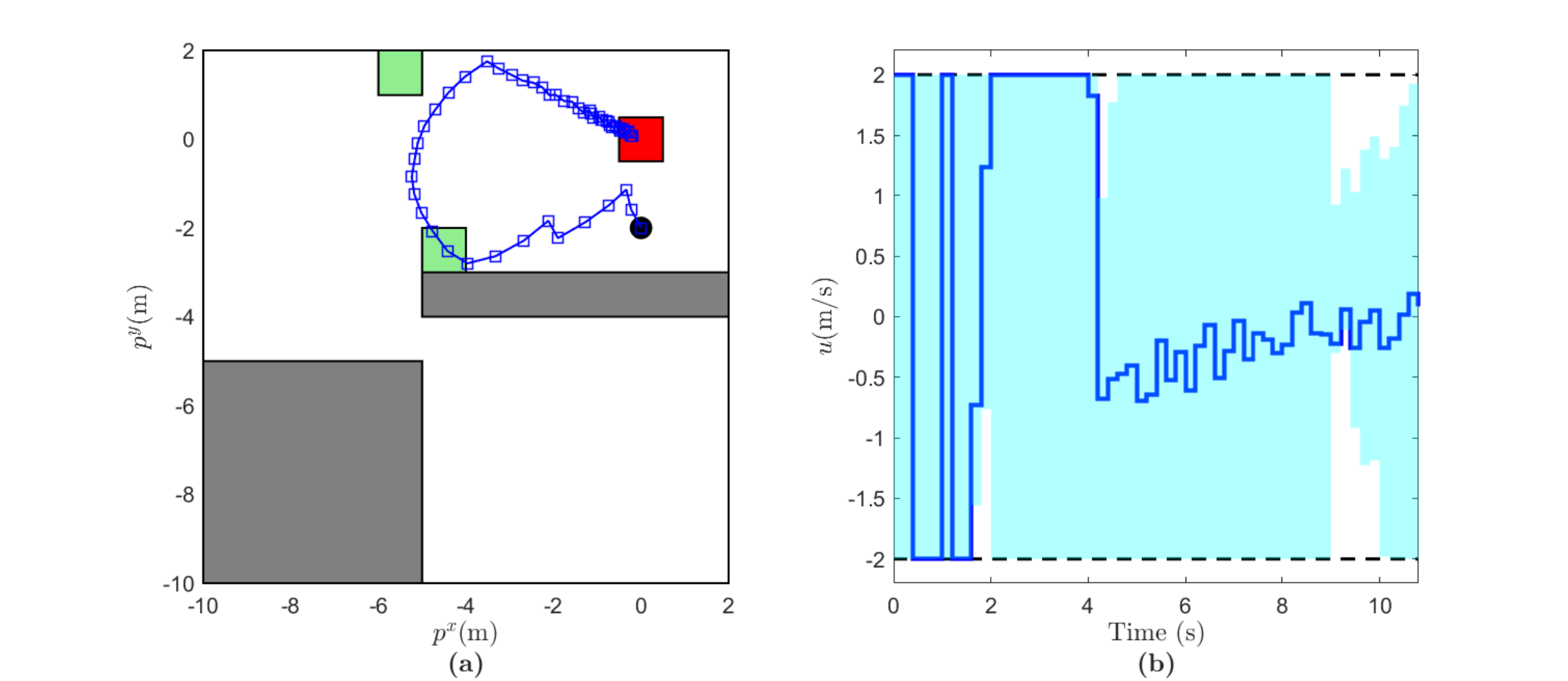}
	\caption{\footnotesize Trajectories starting from the initial states $[0,-2]^T$: (a) state trajectory; (b) control trajectory.}
	\label{Fig:Exam1_case3}
\end{figure}

\begin{figure*}
\centering
\includegraphics[width=0.81\textwidth]{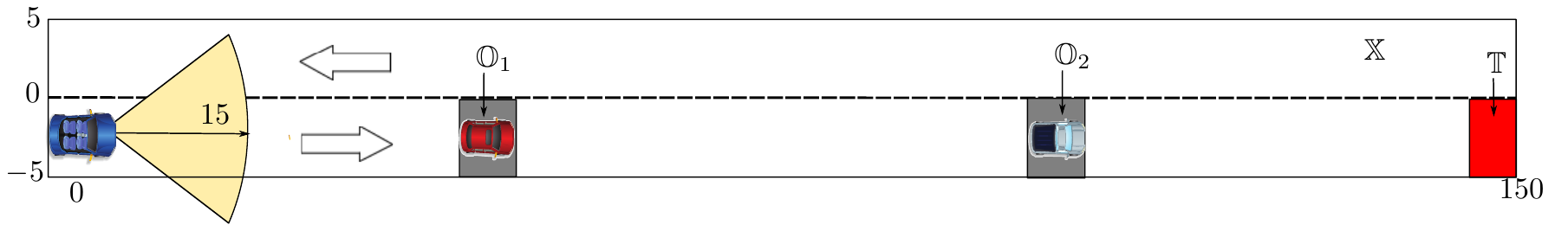}
	\caption{\footnotesize  Scenario illustration of Example 2: an automated vehicle plans to reach a target set $\mathbb{T}$ but with some unknown broken vehicles on the road.}
	\label{Fig:example2_scenariopdf}
\end{figure*}

Recall the system  $\mathsf{CS}$ (\ref{Eq:disConSys}).  Let the set of the observations be $\mathcal{O}=\{a_1,a_2,a_3,a_4,a_5,a_6\}$ and, if $x\in \mathbb{X}$, we define the observation function as
\begin{eqnarray}\label{Eq:Obdoubleintegrator}
g(x)=\begin{cases}
       \{a_1,a_2\}, & \mbox{if } x\in \mathbb{X} \cap \mathbb{O}_1, \\
       \{a_1,a_3\}, & \mbox{if } x\in \mathbb{X} \cap \mathbb{O}_2, \\
       \{a_1,a_4\}, & \mbox{if } x\in \mathbb{X} \cap \mathbb{A}, \\
       \{a_1,a_5\}, & \mbox{if } x\in \mathbb{X} \cap \mathbb{B}, \\
       \{a_1,a_6\}, & \mbox{if } x\in \mathbb{X} \cap \mathbb{T} \\
       \{a_1\}, & \mbox{otherwise}.
     \end{cases}
\end{eqnarray}

As shown in Remark~\ref{SSec:disConSys}, we can rewrite the system (\ref{Eq:doubleintegrator}) with the observation function (\ref{Eq:Obdoubleintegrator}) as a controlled transition system with the set of atomic propositions $\mathcal{AP}=\mathcal{O}$ and the labelling function $L=g$.

In \cite{Gol2014}, the specification is to visit the region $\mathbb{A}$ or region $\mathbb{B}$,
and then the target region $\mathbb{T}$, while always avoiding obstacles $\mathbb{O}_1$
and $\mathbb{O}_2$, and staying inside the working space $\mathbb{X}$. This specification can be expressed as a co-safe LTL formula $\varphi'= ((a_1 \wedge \neg a_2 \wedge \neg a_3) \mathsf{U} a_6) \wedge (\neg a_6 \mathsf{U} (a_4 \vee a_5) )$. Here, we extend the specification  to be to visit region $\mathbb{A}$ or region $\mathbb{B}$,
and then visit and \emph{always} stay inside the target region $\mathbb{T}$, while always avoiding obstacles $\mathbb{O}_1$
and $\mathbb{O}_2$, and staying inside working space $\mathbb{X}$. Obviously, this specification cannot be expressed as a  co-safe LTL formula, and thus cannot be handled by the approach in \cite{Gol2014}. We instead express this specification as the LTL formula $\varphi= ((a_1 \wedge \neg a_2 \wedge \neg a_3) \mathsf{U} \Box a_6) \wedge (\neg a_6 \mathsf{U} (a_4 \vee a_5) )$.  We will show that our approach can handle such non-co-safe LTL formula. By computing inner approximations of the controlled reachable sets, we can construct the controlled TLT of $\varphi$ and then use Algorithm 3 to synthesize controllers online. The constructed controlled TLT for $\varphi$  is shown in Fig.~\ref{Fig:example1_scenariopdf}(b). Similar to \cite{Gol2014}, we choose three different initial states, for each of which the state trajectories and the control trajectories are shown in Figs.~\ref{Fig:Exam1_case1}--\ref{Fig:Exam1_case3}. We can see that in Figs.~\ref{Fig:Exam1_case1}--\ref{Fig:Exam1_case3} (a), all state trajectories satisfy the required specification $\varphi$. The black dots are the initial state. In this example, the target region $\mathbb{T}$ is a RCIS.   After entering $\mathbb{T}$, the states stay there by using the controllers that ensure  robust invariance. In Figs.~\ref{Fig:Exam1_case1}--\ref{Fig:Exam1_case3} (b),  the dashed lines denote the control bounds, the cyan regions represent the synthesized control sets  in Algorithm 3, and the blue lines are the implemented control inputs.

\subsection{Online Specification Update}
This example will show how the specification can be updated online when using our approach. As shown in Fig.~\ref{Fig:example2_scenariopdf}, we consider a scenario where an automated vehicle plans to move to a target set $\mathbb{T}$ but with some unknown obstacles on the road. The sensing region of the vehicle is limited. We use a single integrator model with a sample period of $1$ second to model the dynamics of the vehicle:
\begin{eqnarray}\label{Eq:singleintegrator}
x_{k+1}=\left[\begin{array}{ccccccc}
1 & 0 \\
0 & 1
\end{array}\right]x_k+ \left[\begin{array}{ccccccc}
1  & 0 \\
0  & 1
\end{array}\right]u_k+w_k.
\end{eqnarray}
The working space is $\mathbb{X}=\{z\in \mathbb{R}^2\mid [0, -5]^T \leq z\leq [150, 5]^T\}$, the control  constraint set is $\mathbb{U}=\{z\in \mathbb{R}^2\mid [-2, -0.5]^T \leq z\leq [2, 0.5]^T \}$, the disturbance set is  $\mathbb{W}=\{z\in \mathbb{R}^2\mid [-0.1, -0.1]^T \leq z\leq [0.1, 0.1]^T\}$, and the target region is  $\mathbb{T}=\{z\in \mathbb{R}^2\mid [145, -5]^T \leq z\leq [150, 0]^T\}$. We assume that $\mathbb{X}$, $\mathbb{U}$, and $\mathbb{W}$ are known \emph{a priori} to the vehicle and the vehicle should move along the lane with the right direction unless lane change is necessary.
In Fig.~\ref{Fig:example2_scenariopdf}, there are two broken vehicles in the sets $\mathbb{O}_1=\{z\in \mathbb{R}^2\mid [40, -5]^T \leq z\leq [45, 0]^T\}$ and $\mathbb{O}_2=\{z\in \mathbb{R}^2\mid [100, -5]^T \leq z\leq [105, 0]^T\}$.   We assume that $\mathbb{O}_1$ and $\mathbb{O}_2$ are unknown to the vehicle at the beginning. As long as the vehicle can sense them, they are known to  the vehicle.

\begin{figure}
\centering
 \includegraphics[width=0.31\textwidth]{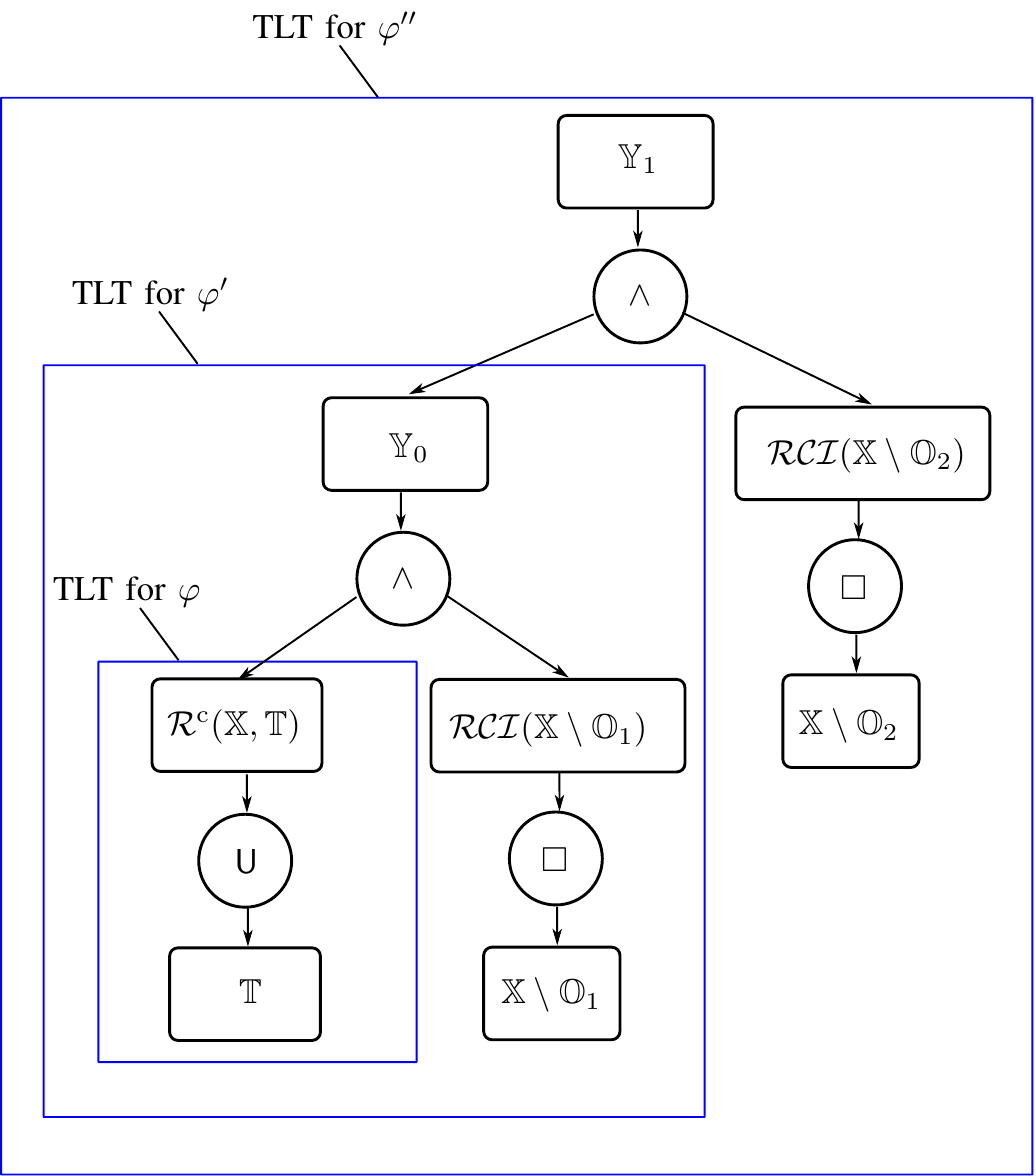}
	\caption{\footnotesize The controlled TLT for the LTL formulae  in Example 2, where $\varphi= a_1 \mathsf{U} a_2$, $\varphi'= \varphi \wedge (\Box \neg a_3)$, $\varphi''= \varphi' \wedge (\Box \neg a_4)$, $\mathbb{Y}_0=\mathcal{R}^{\rm c}(\mathbb{X},\mathbb{T})\cap \mathcal{RCI}(\mathbb{X}\setminus\mathbb{O}_1)$, and $\mathbb{Y}_1=\mathcal{R}^{\rm c}(\mathbb{X},\mathbb{T})\cap \mathcal{RCI}(\mathbb{X}\setminus\mathbb{O}_1)\cap \mathcal{RCI}(\mathbb{X}\setminus\mathbb{O}_2)$.}
	\label{Fig:example2_TLTpdf}
\end{figure}

\begin{figure*}
\hspace{-0.3cm}
	\includegraphics[width=0.95\textwidth]{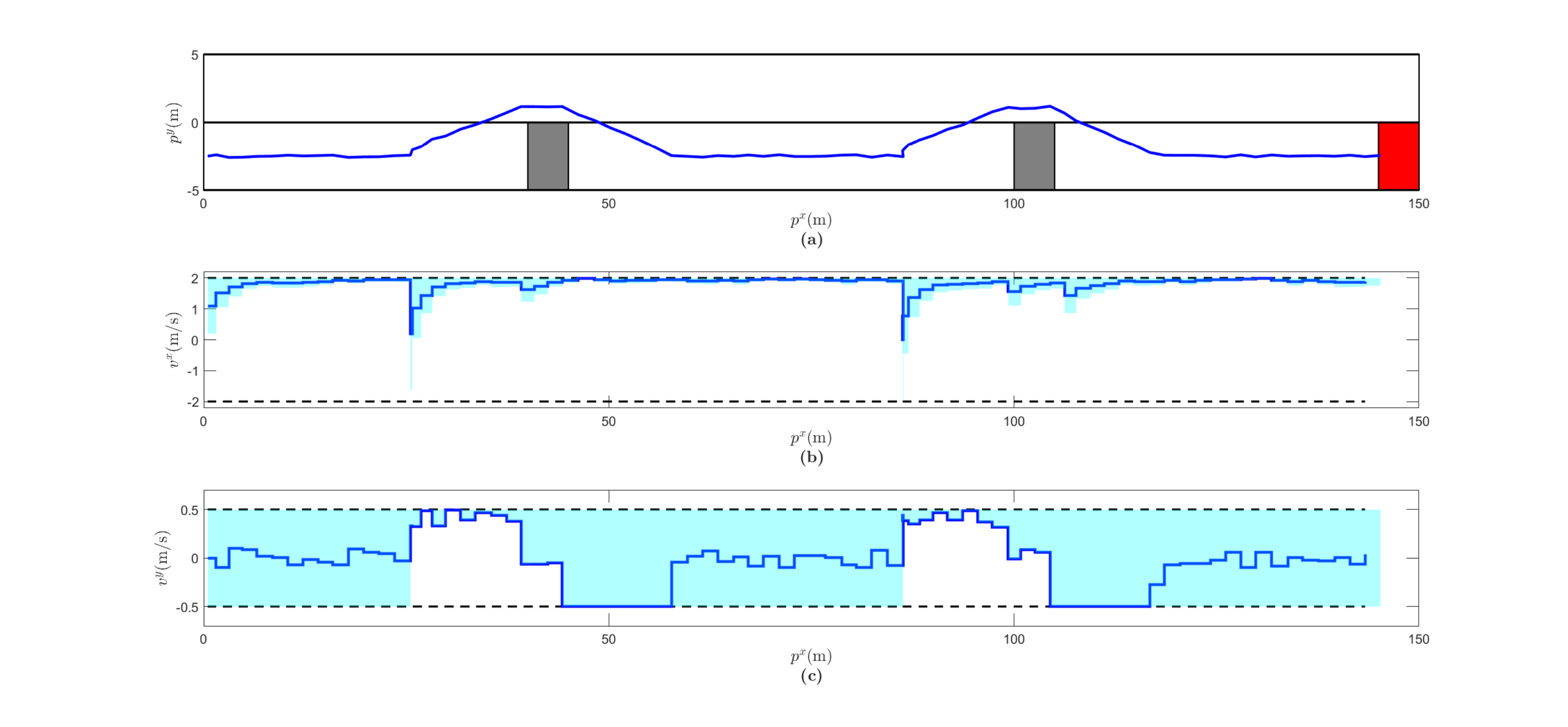}
	\caption{\footnotesize  Trajectories for 1 realization of disturbance trajectories:  (a) state trajectories; (b) control trajectories of $x$-axis; (c) control trajectories of $y$-axis.}
	\label{Fig:Exam2_onerealization}
\end{figure*}

\begin{figure*}
\hspace{-0.3cm}
	\includegraphics[width=0.95\textwidth]{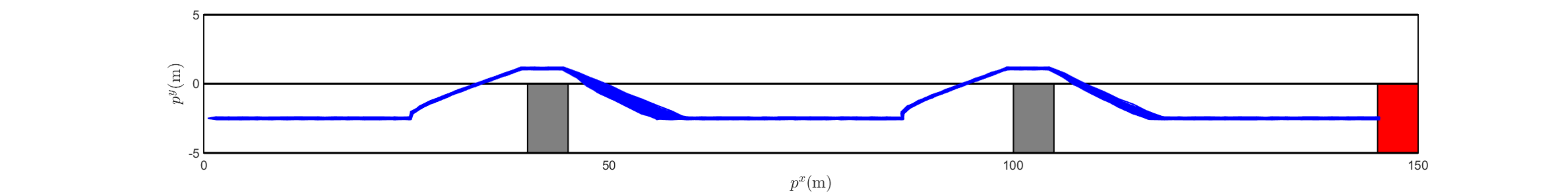}
	\caption{\footnotesize  State trajectories for 100 realizations of disturbance trajectories.}
	\label{Fig:Exam2_100realization}
\end{figure*}

Let the initial state be $x_0=[0.5, -2.5]^T$ and the sensing limitation is $15$. At time step $k=0$, the set of observations is $\mathcal{O}=\{a_1,a_2\}$ and if $x\in \mathbb{X}$, we define the observation function as
\begin{eqnarray*}\label{Eq:Obsingleintegrator1}
g(x)=\begin{cases}
       \{a_1,a_2\}, & \mbox{if } x\in \mathbb{X} \cap \mathbb{T}, \\
       \{a_1\}, & \mbox{otherwise}.
     \end{cases}
\end{eqnarray*}
The initial specification can be expressed as an LTL $\varphi= a_1 \mathsf{U} a_2$. By constructing the controlled TLT of $\varphi$ shown in Fig.~\ref{Fig:example2_TLTpdf} and implementing Algorithm 3, we obtain one realization as shown in Fig.~\ref{Fig:Exam2_onerealization}. We can see that the vehicle keeps moving straightforward until it senses the obstacle $\mathbb{O}_1$ at $[25.5, -2.4]^T$.

When the vehicle can sense $\mathbb{O}_1$, a new observation $a_3$ with $a_3\neq a_1$ and  $a_3\neq a_2$ is added to the set $\mathcal{O}$, which becomes $\mathcal{O}=\{a_1,a_2,a_3\}$.  If $x\in \mathbb{X}$, we update the observation function as
\begin{eqnarray*}\label{Eq:Obsingleintegrator2}
g(x)=\begin{cases}
       \{a_1,a_2\}, & \mbox{if } x\in \mathbb{X} \cap \mathbb{T}, \\
       \{a_1,a_3\}, & \mbox{if } x\in \mathbb{X} \cap \mathbb{O}_1, \\
       \{a_1\}, & \mbox{otherwise}.
     \end{cases}
\end{eqnarray*}
To avoid $\mathbb{O}_1$, the new specification is changed to be $\varphi'= \varphi \wedge (\Box \neg a_3)$. We can construct the TLT of $\varphi'$ based on that of $\varphi$, which is shown in Fig.~\ref{Fig:example2_TLTpdf}, and then continue to implement Algorithm 3. We can see that the vehicle changes lane from $[25.5, -2.4]^T$ and quickly merges back after overtaking $\mathbb{O}_1$. The trajectories are shown in Fig.~\ref{Fig:Exam2_onerealization}. The vehicle is under the control with respect to $\varphi'$ until it can sense  $\mathbb{O}_2$ at $[86.3, -2.5]^T$.

Similarly, when the vehicle can sense $\mathbb{O}_2$, we update $\mathcal{O}=\{a_1,a_2,a_3,a_4\}$ and the observation function as  if $x\in \mathbb{X}$,
\begin{eqnarray*}\label{Eq:Obsingleintegrator3}
g(x)=\begin{cases}
       \{a_1,a_2\}, & \mbox{if } x\in \mathbb{X} \cap \mathbb{T}, \\
       \{a_1,a_3\}, & \mbox{if } x\in \mathbb{X} \cap \mathbb{O}_1, \\
       \{a_1,a_4\}, & \mbox{if } x\in \mathbb{X} \cap \mathbb{O}_2, \\
       \{a_1\}, & \mbox{otherwise}.
     \end{cases}
\end{eqnarray*}
To avoid $\mathbb{O}_2$, the new specification is changed  to be $\varphi''= \varphi'\wedge (\Box \neg a_4)$. We can construct the TLT of $\varphi''$ based on that of $\varphi'$, which is shown in Fig.~\ref{Fig:example2_TLTpdf},  and then continue to implement Algorithm 3. We can see that the vehicle changes lane from $[86.3, -2.5]^T$ and quickly merges back after overtaking $\mathbb{O}_2$. Under the control with respect to $\varphi''$, the vehicle finally reaches the target set $\mathbb{T}$.

Fig.~\ref{Fig:Exam2_onerealization} (a) shows the state trajectories, from which we can see that the whole specification is completed.  Fig.~\ref{Fig:Exam2_onerealization} (b)--(c) show the corresponding control inputs, where the dashed lines denote the control bounds.
The cyan regions represent the synthesized control
sets and the blue lines are the control trajectories. Furthermore, we repeat the above process for  100 realizations of the disturbance trajectories. The state trajectories for such 100 realizations are shown in Fig.~\ref{Fig:Exam2_100realization}.

We remark that in this example,  the  control inputs are chosen to push the state to move down along the TLT as fast as possible. In detail, if the state $x_k$ is the $i$-step reachable set in the set node $\mathcal{R}^{\rm c}(\mathbb{X},\mathbb{T})$, we can generate a smaller control set from which the control input can push the state to the $(i-1)$-step reachable set. That is what we can see from Fig.~\ref{Fig:Exam2_onerealization}, where almost all control inputs in the synthesized control sets along $x$-axis are positive.

\section{Conclusions and Future Work}\label{Sec:Conclusion}

We have studied  LTL model checking and control synthesis for discrete-time uncertain systems. Quite unlike automaton-based methods, our solutions build on the connection between LTL formulae and TLT structures via reachability analysis. For a transition system and an LTL formula, we have proved that the TLTs provide an
underapproximation or overapproximation for the LTL via minimal and maximal reachability analysis, respectively.  We have provided  sufficient
conditions and  necessary conditions to the model checking problem.
   For a controlled transition system and an LTL formula, we have shown
that the TLT is an underapproximation for the LTL
formula and thereby proposed an online control synthesis
algorithm, under which a set of feasible control inputs is
generated at each time step. We have proved
that this algorithm is recursively feasible. We have also illustrated the effectiveness of the proposed methods  through several examples.

Future work includes the extension of TLTs to handle other
general  specifications (e.g., CTL$^*$) and a broad experimental
evaluation of our approach.

% how to efficiently choose the control input from the control set,

\section*{Acknowledgements}
The authors are grateful to Dr. Xiaoqiang Ren (Shanghai University) and Mr. Hosein Hasanbeig (University of Oxford) for helpful discussions and feedback.

\begin{figure*}
\centering
\begin{subfigure}{.34\textwidth}
\includegraphics[width=5.3cm,height=3.6cm,keepaspectratio]{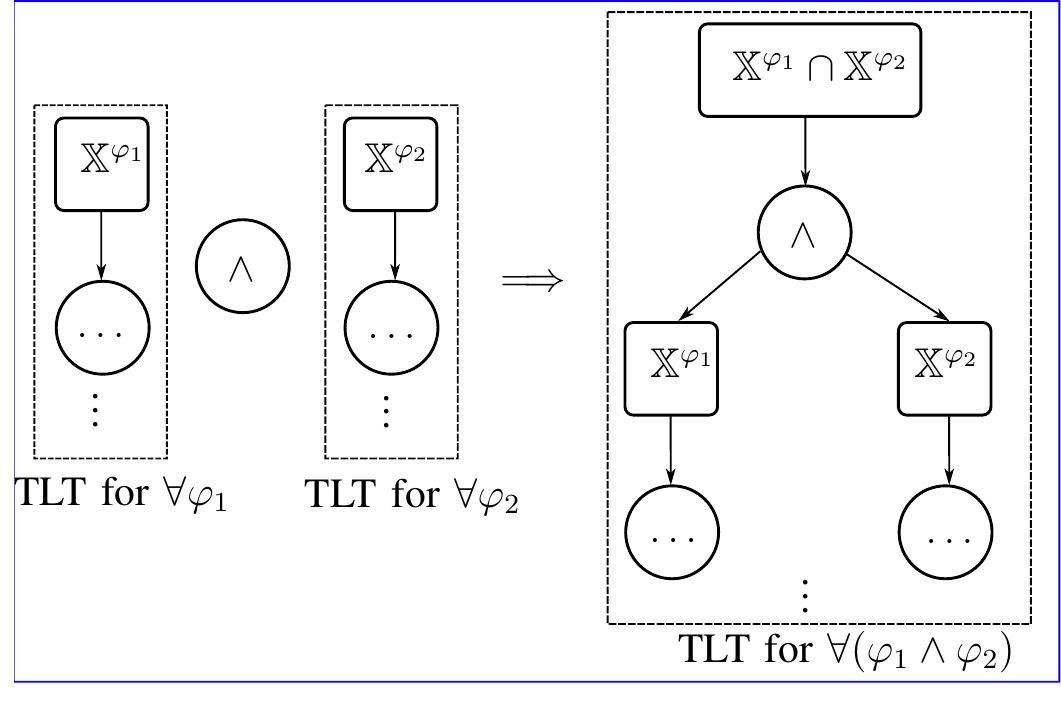}
  \caption{}
\end{subfigure}
\begin{subfigure}{.34\textwidth}
\includegraphics[width=5.3cm,height=3.6cm,keepaspectratio]{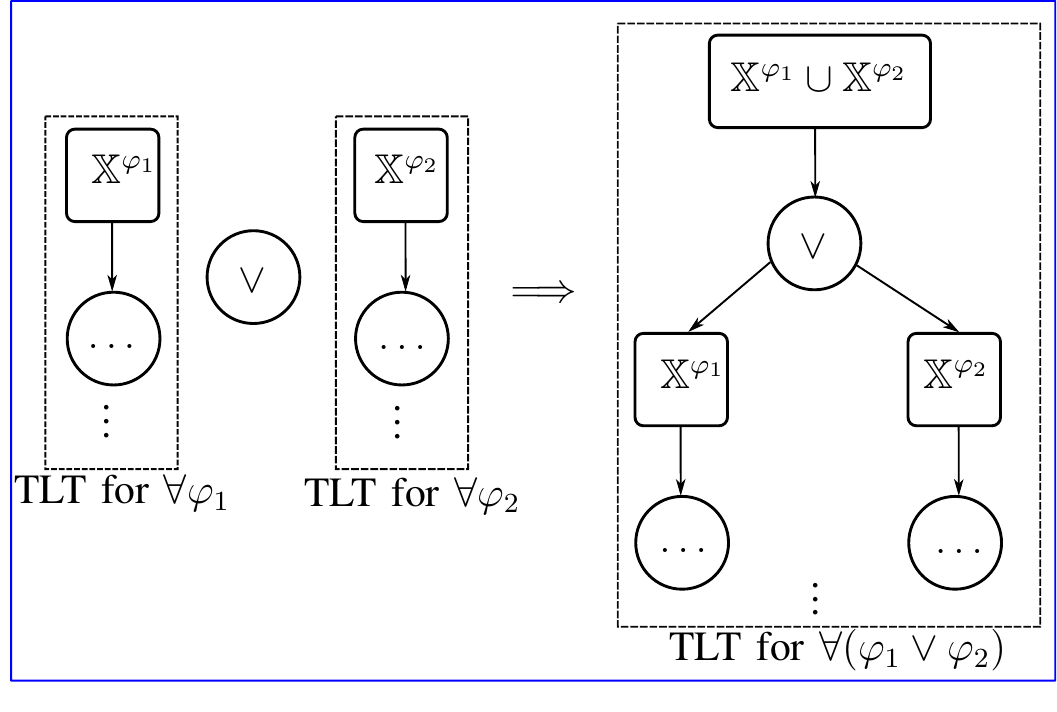}
  \caption{}
\end{subfigure}
\begin{subfigure}{.28\textwidth}
\includegraphics[width=5.3cm,height=3.6cm,keepaspectratio]{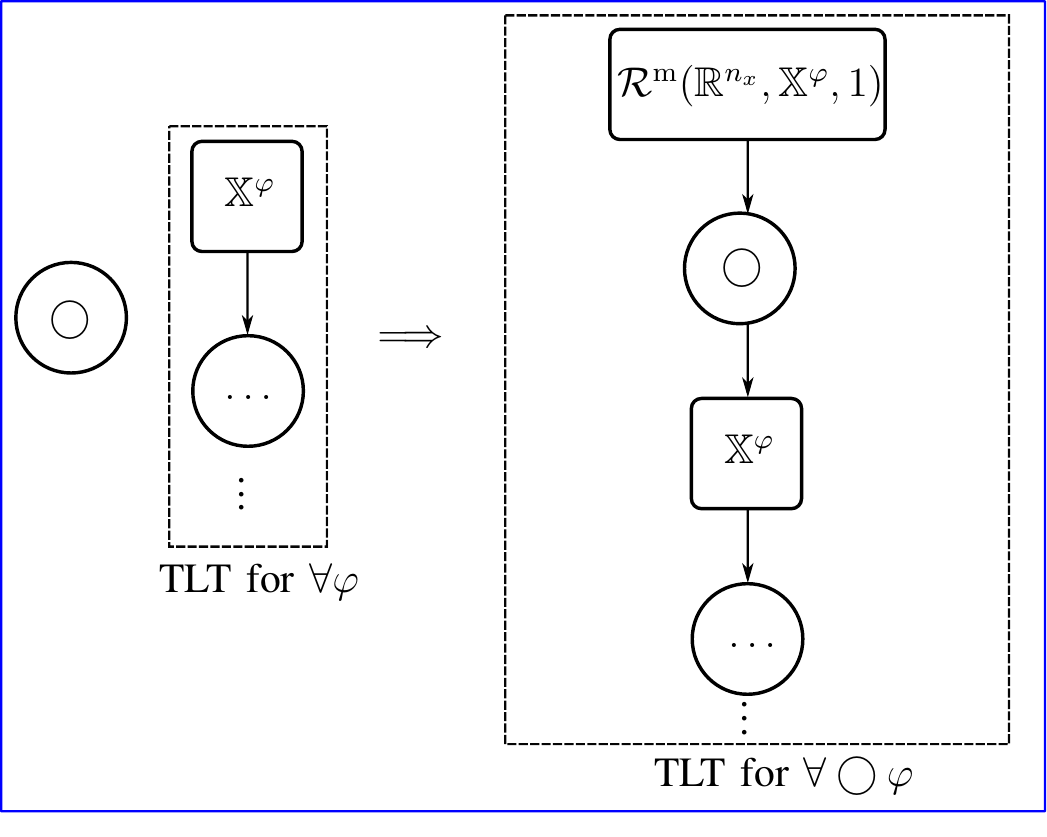}
  \caption{}
\end{subfigure}

\hspace{-0.6cm}
\begin{subfigure}{.45\textwidth}
\includegraphics[width=10cm,height=7cm,keepaspectratio]{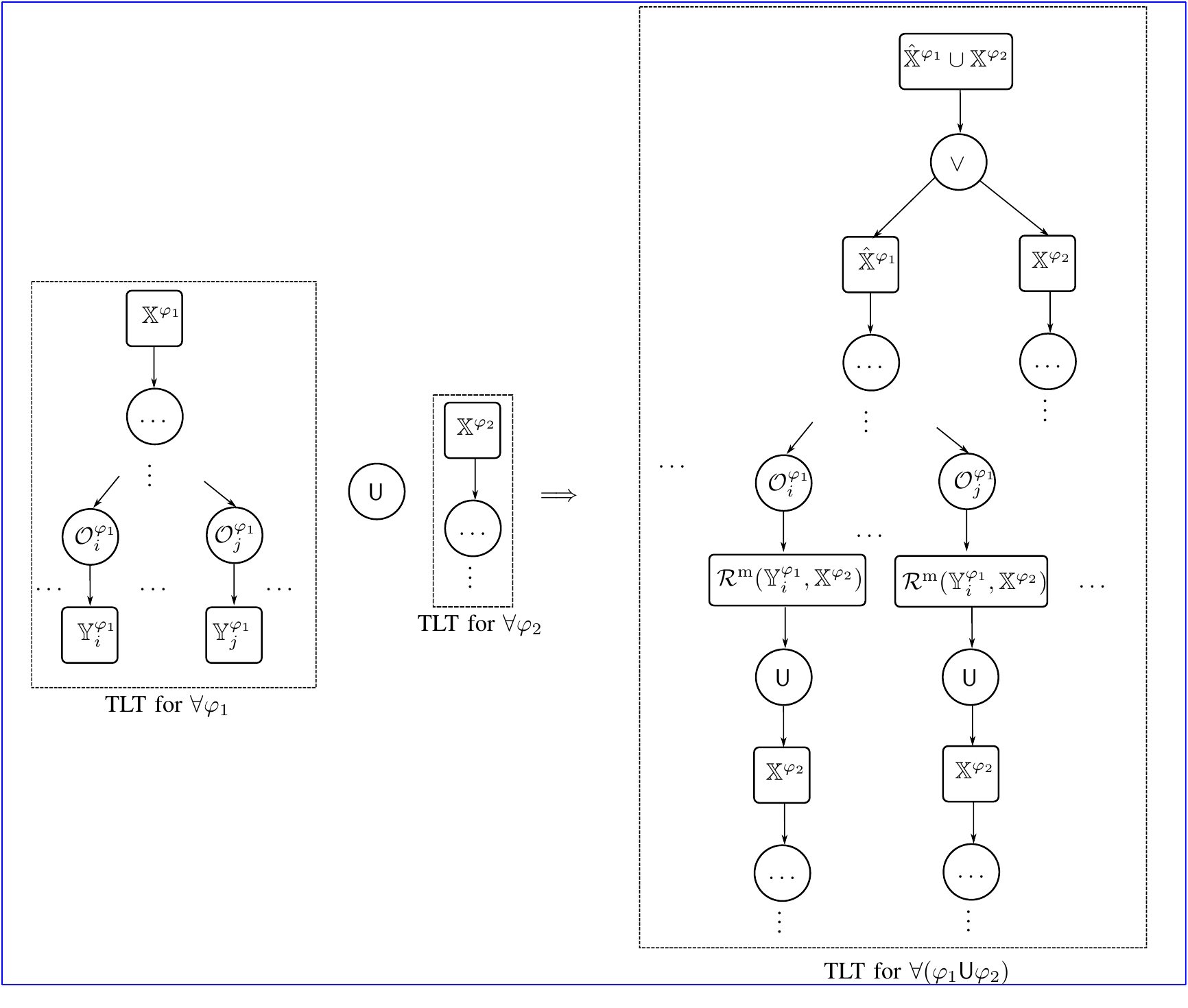}
  \caption{}
\end{subfigure}
\hspace{0.6cm}
\begin{subfigure}{0.42\textwidth}
\includegraphics[width=10cm,height=7cm,keepaspectratio]{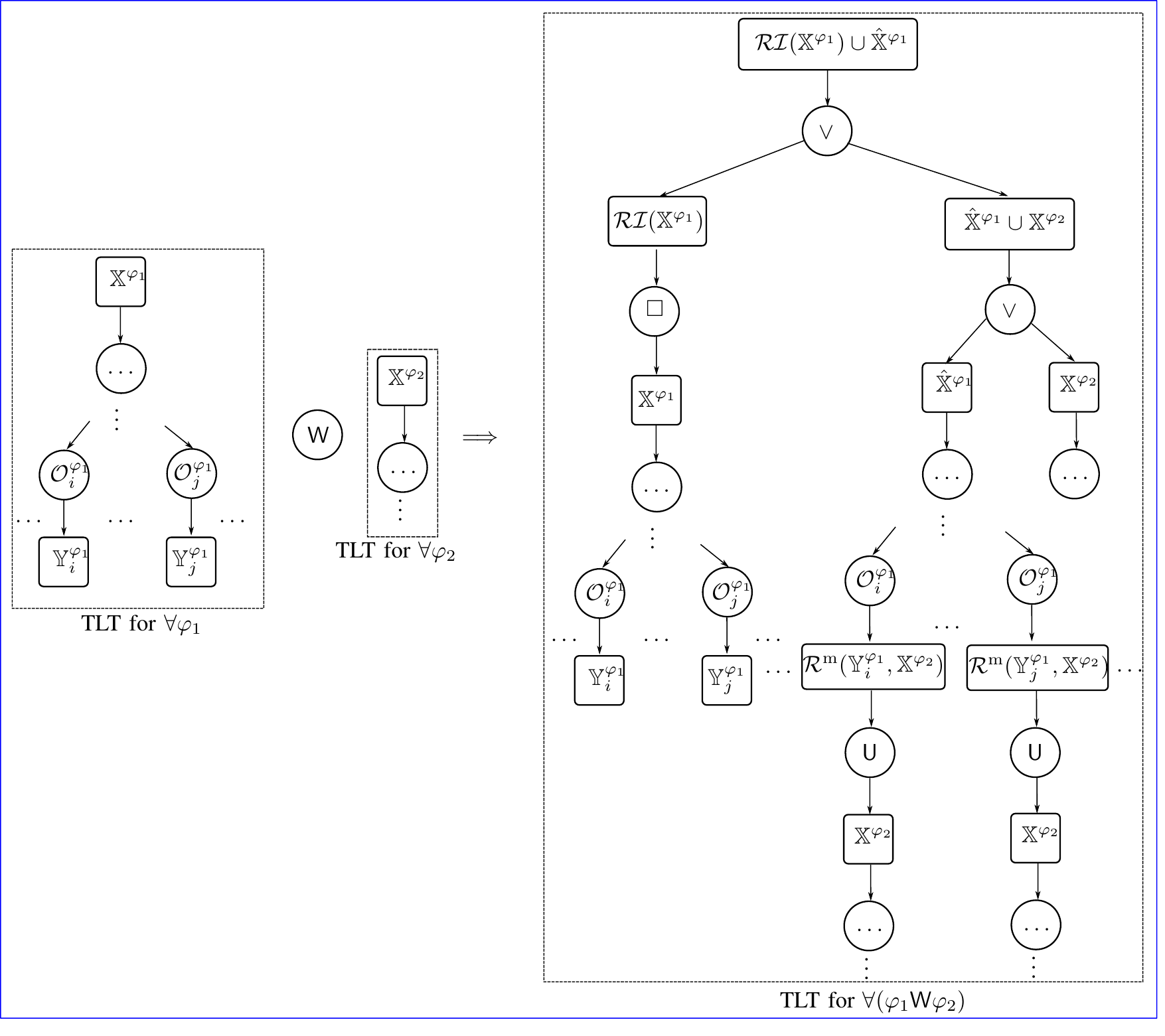}
  \caption{}
\end{subfigure}
	\caption{\footnotesize The TLT construction: (a) $\forall (\varphi_1 \wedge \varphi_2)$; (b) $\forall (\varphi_1 \vee \varphi_2)$; (c) $\forall \bigcirc \varphi$; (d) $\forall (\varphi_1 \mathsf{U} \varphi_2)$; (e) $\forall (\varphi_1 \mathsf{W} \varphi_2)$. Here, the circles denote the operator nodes and the rectangles denote the set nodes.}
	\label{Fig:TLTCons}
\end{figure*}

\section*{Appendix A. Proof of Theorem~\ref{The:UniTLTLTL}}
The whole proof is divided into two parts: the first part shows how to construct  a TLT from the formula $\forall \varphi$ by means of the reachability  operators $\mathcal{R}^{\rm m}$ and $\mathcal{RI}$, while the second part shows that such TLT is an underapproximation for $\varphi$.

\textbf{Construction}: We follow three steps to construct a TLT.
%from $\forall \varphi$ via the reachability  operators $\mathcal{R}^{\rm m}$ and $\mathcal{RI}$.

\emph{Step 1: rewrite the given LTL in the weak-until positive normal form.}
From \cite{Baier2008}, each LTL formula has an equivalent LTL formula in the weak-until positive normal form,
%where negations only occur over atomic propositions,
which can be inductively defined as
\begin{eqnarray}\label{weakuntilLTLdef}
&&\varphi ::= {\rm{true}} \mid {\rm{false}} \mid a \mid \neg a \mid \varphi_1 \wedge \varphi_2 \mid \varphi_1\vee \varphi_2 \mid \bigcirc \varphi \nonumber \\
&&\hspace{1cm} \mid  \varphi_1 \mathsf{U} \varphi_2 \mid  \varphi_1 \mathsf{W} \varphi_2.
\end{eqnarray}

\emph{Step 2:  for each atomic proposition $a\in\mathcal{AP}$, construct the TLT with only a single set node from $\forall a$ or $\forall \neg a$.} In detail, the set node for $\forall a$ is $L^{-1}(a)=\{x\in \mathbb{S}\mid a\in L(x)\}$ while the set node for $\forall \neg a$ is $\mathbb{S}\setminus L^{-1}(a)$. In addition, the TLT for $\forall {\rm{true}}$ (or $\forall{\rm{false}}$) also has a single set node, which is $\mathbb{S}$ (or $\emptyset$).

\emph{Step 3: based on Step 2, follow the induction rule  to construct the TLT for any LTL formula in the weak-until positive normal form.} More specifically, we will show that given the LTL formulae $\varphi$, $\varphi_1$, and $\varphi_2$ in the weak-until positive normal form, if the  TLTs can be constructed from $\forall \varphi$, $\forall \varphi_1$, and $\forall \varphi_2$, respectively, then the TLTs can be thereby constructed from the formulae $\forall (\varphi_1 \wedge \varphi_2)$, $\forall (\varphi_1\vee \varphi_2)$, $\forall \bigcirc \varphi$, $\forall (\varphi_1 \mathsf{U} \varphi_2)$, and $\forall (\varphi_1 \mathsf{W} \varphi_2)$, respectively.

For $\forall (\varphi_1 \wedge \varphi_2)$ (or $\forall(\varphi_1\vee \varphi_2)$), we construct the TLT  by connecting the root nodes of the TLTs for  $\forall \varphi_1$ and $\forall \varphi_2$ through the operator  $\wedge$ (or $\vee$) and taking the intersection (or union) of two root nodes, as shown in Fig.~\ref{Fig:TLTCons}(a)--(b).  For $\forall \bigcirc \varphi$, we denote by $\mathbb{X}^{\varphi}$ the root node of the TLT for $\forall \varphi$ and then construct the TLT by adding a new set node $\mathcal{R}^{\rm{m}}(\mathbb{S},\mathbb{X}^{\varphi},1)$ to be the parent of $\mathbb{X}^{\varphi}$ and connecting them through the operator $\bigcirc$, as shown in Fig.~\ref{Fig:TLTCons}(c).

For $\forall (\varphi_1 \mathsf{U} \varphi_2)$, the TLT construction is as follows.   Denote by $\{(\mathbb{Y}^{\varphi_1}_i, \mathcal{O}^{\varphi_1}_i)\}_{i=1}^{N^{\varphi_1}}$ all the pairs comprising a leaf node and its corresponding parent in the TLT of $\forall \varphi_1$, where $N^{\varphi_1}$ is the number of the leaf nodes. Here, $\mathbb{Y}^{\varphi_1}_i$ is the $i$th leaf node and $\mathcal{O}^{\varphi_1}_i$ is its parent.   Denote by  $\mathbb{X}^{\varphi_2}$ the root node of TLT for $\forall\varphi_2$.
 We first change each leaf node $\mathbb{Y}^{\varphi_1}_i$ to $\mathcal{R}^{\rm{m}}(\mathbb{Y}^{\varphi_1}_i,\mathbb{X}^{\varphi_2})\setminus \mathbb{X}^{\varphi_2} $. We then update the new tree for $\forall \varphi_1$ from the leaf node to the root node according to the definition of the operators. After that, we take $N^{\varphi_1}$ copies of the TLT of $\varphi_2$. We set the root node of each copy as the child of each new leaf node, respectively, and connect them trough the operator $\mathsf{U}$. Finally, we have one more copy of the TLT of $\forall \varphi_2$ and connect this copy and the new tree  trough the disjunction $\vee$. An illustrative diagram is given in Fig.~\ref{Fig:TLTCons}(d).

For the fragment $\forall (\varphi_1 \mathsf{W} \varphi_2)$, we first recall that $\varphi_1 \mathsf{W}\varphi_2 = \varphi_1 \mathsf{U} \varphi_2 \vee \Box \varphi_1$. Let $\varphi'= \varphi_1 \mathsf{U} \varphi_2$ and $\varphi''=\Box \varphi_1$. Denote by  $\mathbb{X}^{\varphi_1}$ the root node of the TLT for $\forall \varphi_1$.
We first  construct the  TLT of $\forall \varphi'$ as described above.  Second, we further construct the TLT of $\forall  \varphi''$ with by adding a new node $\mathcal{RI}(\mathbb{X}^{\varphi_1})$ as the parent of $\mathbb{X}^{\varphi_1}$ and connecting them through $\Box$. Then, we
 construct the TLT of $\forall (\varphi'\vee\varphi'')$. An illustrative diagram is given in Fig.~\ref{Fig:TLTCons}(e).

\textbf{Underapproximation}:  First, it is very easy to verify that the constructed TLT above with a single set node $L^{-1}(a)$ (or $\mathbb{S}\setminus L^{-1}(a)$ or ) for $\forall a$ (or $\forall \neg a$ or $\mathbb{S}$ or $\emptyset$) is an underapproximation for $a\in\mathcal{AP}$ (or $\neg a$ or $\forall {\rm{true}}$ or $\forall{\rm{false}}$) and the underapproximation relation in these cases  is also tight.

%For each atomic proposition $a\in\mathcal{AP}$ and its negation $\neg a$, it is easy to verify that the constructed TLT above with a single set node $L^{-1}(a)$ (or $\mathbb{S}\setminus L^{-1}(a)$) for $\forall a$ (or $\forall \neg a$) collects all the initial states from which the trajectories satisfy $a$ (or $\neg a$), which thus is an underapproximation of $a$ (or $\neg a$).
% Following a similar idea, we can show that
% the corresponding TLT of $\forall {\rm{true}}$ (or $\forall{\rm{false}}$) is also an underapproximation of ${\rm{true}}$ (or ${\rm{false}}$).
Next we also follow the induction rule to show that the constructed TLT from $\forall \varphi$ is an underapproximation for  $\varphi$. Consider LTL formulae $\varphi$, $\varphi_1$, and $\varphi_2$. We will show that  if  the constructed TLTs of $\forall \varphi$, $\forall \varphi_1$, and $\forall\varphi_2$ are the underapproximations of $\varphi$, $\varphi_1$, and $\varphi_2$, respectively, then the  TLTs constructed above  for the formulae $\forall (\varphi_1 \wedge \varphi_2)$, $\forall (\varphi_1\vee \varphi_2)$, $\forall \bigcirc \varphi$, $\forall (\varphi_1 \mathsf{U} \varphi_2)$, and $\forall (\varphi_1 \mathsf{W} \varphi_2)$  are the underapproximations of $\varphi_1 \wedge \varphi_2$, $\varphi_1\vee \varphi_2$, $\bigcirc \varphi$, $\varphi_1 \mathsf{U} \varphi_2$, and $\varphi_1 \mathsf{W} \varphi_2$, respectively.

According to the set operation (intersection or union) or the definition of one-step minimal reachable set, it is easy to verify that the constructed TLT for  $\forall (\varphi_1 \wedge \varphi_2)$ (or $\forall (\varphi_1\vee \varphi_2)$)  or $\forall \bigcirc \varphi$ is an underapproximation for $\varphi_1 \wedge \varphi_2$ (or $\varphi_1\vee \varphi_2$) or  $\bigcirc \varphi$ if the TLTs of  $\forall \varphi_1$ and $\forall\varphi_2$, and $\forall \varphi$ are underapproximations $\varphi$, $\varphi_1$, and $\varphi_2$, respectively.

Let us consider $\varphi_1\mathsf{U} \varphi_2$.  Assume that  a trajectory $\bm{p}$ satisfies  the TLT of  $\forall (\varphi_1\mathsf{U} \varphi_2)$. Recall the construction of the TLT of  $\forall (\varphi_1\mathsf{U} \varphi_2)$ from  $\forall \varphi_1$ and $\forall \varphi_2$. According to the definition of minimal reachable set, we have (1) $\bm{p}$ satisfies  the TLT of $\forall \varphi_2$; \emph{or} (2) there exists that $j\in \mathbb{N}$ such that
$\bm{p}[j..]$ satisfies  the TLT of $\forall \varphi_2$ and for all $i\in \mathbb{N}_{[0,j-1]}$, the trajectory $p[i..]$ satisfies the the TLT of $\forall \varphi_1$. Under the assumption that the TLTs of $\forall \varphi_1$ and $\forall \varphi_2$ are the underapproximations of $\varphi_1$ and $\varphi_2$, respectively, we have that there exists $j\in \mathbb{N}$ such that $\bm{p}[j..] \vDash \varphi_2$ and for all $i\in \mathbb{N}_{[0,j-1]}$, $p[i..]\vDash \varphi_1$, which implies that $\bm{p}\vDash \varphi_1\mathsf{U} \varphi_2$. Thus, the TLT of  $\forall (\varphi_1\mathsf{U} \varphi_2)$ is an approximation of $\varphi_1\mathsf{U} \varphi_2$.

Recall that $\varphi_1 \mathsf{W}\varphi_2 = \varphi_1 \mathsf{U} \varphi_2 \vee \Box \varphi_1$. Following the proofs for until operator $\mathsf{U}$ and the disjunction $\vee$ and the definition of the robust invariant set, it yields that the constructed TLT of  $\forall (\varphi_1\mathsf{W} \varphi_2)$ is an underapproximation of $\varphi_1\mathsf{W} \varphi_2$.

The proof is complete. \qed

\end{document}